%% file: main.tex
\documentclass[sigconf,natbib,screen=true]{acmart}
\input{preamble/packages}
\input{preamble/definitions}
\input{preamble/meta}
\input{preamble/authors}

\title[Group Membership Bias]{The Impact of Group Membership Bias on the Quality and Fairness of Exposure in Ranking}

\allowdisplaybreaks

\begin{document}
%\fancyhead{} 

\begin{abstract}
When learning to rank from user interactions, search and recommender systems must address biases in user behavior to provide a high-quality ranking.
One type of bias that has recently been studied in the ranking literature is when sensitive attributes, such as gender, have an impact on a user's judgment about an item's utility.
For example, in a search for an expertise area, some users may be biased towards clicking on male candidates over female candidates.
We call this type of bias \emph{group membership bias}.% or group bias for short.

Increasingly, we seek rankings that are fair to individuals and sensitive groups.
Merit-based fairness measures rely on the estimated utility of the items.
With group membership bias, the utility of the sensitive groups is under-estimated, hence, without correcting for this bias, a supposedly fair ranking is not truly fair.
In this paper, first, we analyze the impact of group membership bias on ranking quality as well as merit-based fairness metrics and show that group membership bias can hurt both ranking and fairness.
Then, we provide a correction method for group bias that is based on the assumption that the utility score of items in different groups comes from the same distribution.
This assumption has two potential issues of sparsity and equality-instead-of-equity; we use an amortized approach to address these.
We show that our correction method can consistently compensate for the negative impact of group membership bias on ranking quality and fairness metrics.
\vspace*{-1mm}
\end{abstract}

\maketitle

\acresetall

\input{sections/01-introduction}
\input{sections/02-related}
\input{sections/03-groupbias}
\input{sections/04-theory}
\input{sections/05-correction}
\input{sections/06-experiment}
\input{sections/07-conclusion}

\vspace{-0.3em}
\section*{Code and data}
To ensure the reproducibility of the reported results, this work only made use of publicly available data and our experimental implementation can be accessed publicly at 
\url{https://github.com/AliVard/groupbias}.

\vspace{-0.3em}
\begin{acks}
This research was supported by Elsevier, by the Netherlands Organisation for Scientific Research (NWO)
under pro\-ject nrs
612.\-001.\-551, 024.004.022, NWA.1389.20.\-183, and KICH3.LTP.20.006, and by the European Union's Horizon Europe program under grant agreement No. 101070212. All content represents the opinion of the authors, which is not necessarily shared or endorsed by their respective employers and/or sponsors.\looseness=-1
\end{acks}
 
\clearpage

\bibliographystyle{bibliography/ACM-Reference-Format}
\balance
\bibliography{bibliography/references}

\end{document}

%% file: preamble/packages.tex
\PassOptionsToPackage{dvipsnames}{xcolor}

\usepackage{dsfont}
\usepackage{acronym}
\usepackage[noend]{algorithmic}
\usepackage[ruled,vlined,linesnumbered]{algorithm2e}
\usepackage{bbm}
\usepackage{beramono}
\usepackage{bm}
\usepackage{booktabs}
\usepackage[skip=0pt]{caption}
\usepackage[capitalise]{cleveref}
\usepackage[T1]{fontenc}

\usepackage{graphicx}
\usepackage{hyperref}
\usepackage{listings}
\usepackage{multirow}
\usepackage{subcaption}
\usepackage{tabularx}
\usepackage[dvipsnames]{xcolor}
\usepackage{colortbl}
\usepackage{calc}
\usepackage[inline]{enumitem}
\usepackage{numprint}
\usepackage[normalem]{ulem}

\usepackage{balance}
\usepackage{mathrsfs}
\usepackage{mathtools}
\usepackage{xspace}

\usepackage{tikz}
\usetikzlibrary{positioning, fadings,patterns,fit}% To get more advances positioning options
\usetikzlibrary{arrows}% To get more arrow heads

\makeatletter
\newif\ifAC@uppercase@first%
\def\Aclp#1{\AC@uppercase@firsttrue\aclp{#1}\AC@uppercase@firstfalse}%
\def\AC@aclp#1{%
  \ifcsname fn@#1@PL\endcsname%
    \ifAC@uppercase@first%
      \expandafter\expandafter\expandafter\MakeUppercase\csname fn@#1@PL\endcsname%
    \else%
      \csname fn@#1@PL\endcsname%
    \fi%
  \else%
    \AC@acl{#1}s%
  \fi%
}%
\def\Acp#1{\AC@uppercase@firsttrue\acp{#1}\AC@uppercase@firstfalse}%
\def\AC@acp#1{%
  \ifcsname fn@#1@PL\endcsname%
    \ifAC@uppercase@first%
      \expandafter\expandafter\expandafter\MakeUppercase\csname fn@#1@PL\endcsname%
    \else%
      \csname fn@#1@PL\endcsname%
    \fi%
  \else%
    \AC@ac{#1}s%
  \fi%
}%
\def\Acfp#1{\AC@uppercase@firsttrue\acfp{#1}\AC@uppercase@firstfalse}%
\def\AC@acfp#1{%
  \ifcsname fn@#1@PL\endcsname%
    \ifAC@uppercase@first%
      \expandafter\expandafter\expandafter\MakeUppercase\csname fn@#1@PL\endcsname%
    \else%
      \csname fn@#1@PL\endcsname%
    \fi%
  \else%
    \AC@acf{#1}s%
  \fi%
}%
\def\Acsp#1{\AC@uppercase@firsttrue\acsp{#1}\AC@uppercase@firstfalse}%
\def\AC@acsp#1{%
  \ifcsname fn@#1@PL\endcsname%
    \ifAC@uppercase@first%
      \expandafter\expandafter\expandafter\MakeUppercase\csname fn@#1@PL\endcsname%
    \else%
      \csname fn@#1@PL\endcsname%
    \fi%
  \else%
    \AC@acs{#1}s%
  \fi%
}%
\edef\AC@uppercase@write{\string\ifAC@uppercase@first\string\expandafter\string\MakeUppercase\string\fi\space}%
\def\AC@acrodef#1[#2]#3{%
  \@bsphack%
  \protected@write\@auxout{}{%
    \string\newacro{#1}[#2]{\AC@uppercase@write #3}%
  }\@esphack%
}%
\def\Acl#1{\AC@uppercase@firsttrue\acl{#1}\AC@uppercase@firstfalse}
\def\Acf#1{\AC@uppercase@firsttrue\acf{#1}\AC@uppercase@firstfalse}
\def\Acfi#1{\AC@uppercase@firsttrue\acfi{#1}\AC@uppercase@firstfalse}
\def\Ac#1{\AC@uppercase@firsttrue\ac{#1}\AC@uppercase@firstfalse}
\def\Acs#1{\AC@uppercase@firsttrue\acs{#1}\AC@uppercase@firstfalse}
\robustify\Aclp
\robustify\Acfp
\robustify\Acp
\robustify\Acsp
\robustify\Acl
\robustify\Acf
\robustify\Acfi
\robustify\Ac
\robustify\Acs
\makeatother

\def\AC@@acro#1[#2]#3{%
  \ifAC@nolist%
  \else%
  \ifAC@printonlyused%
    \expandafter\ifx\csname acused@#1\endcsname\AC@used%
       \item[\protect\AC@hypertarget{#1}{\acsfont{#2}}] #3%
          \ifAC@withpage%
            \expandafter\ifx\csname r@acro:#1\endcsname\relax%
               \PackageInfo{acronym}{%
                 Acronym #1 used in text but not spelled out in
                 full in text}%
            \else%
               \dotfill\pageref{acro:#1}%
            \fi\\%
          \fi%
    \fi%
 \else%
    \item[\protect\AC@hypertarget{#1}{\acsfont{#2}}] #3%
 \fi%
 \fi%
 \begingroup
    \def\acroextra##1{}%
    \@bsphack
    \protected@write\@auxout{}%
       {\string\newacro{#1}[\string\AC@hyperlink{#1}{#2}]{\AC@uppercase@write #3}}%
    \@esphack
  \endgroup}

%% file: preamble/definitions.tex
\definecolor{ali}{RGB}{0, 150, 0}
\definecolor{mdr}{RGB}{255, 127, 0}

\definecolor{Gray}{gray}{0.8}
\definecolor{Gray2}{gray}{0.93}

\newcolumntype{D}{>{\columncolor{Gray2}}c}
\newcolumntype{B}{>{\itshape\columncolor{white}}c}
\newcolumntype{A}{>{\columncolor{Gray}}c}

\newcommand{\myparagraph}[1]{\vspace*{0.7mm}\emph{\bfseries #1}}

\newcommand{\mExpect}[2]{\ensuremath{\mathbb{E}_{#1}\left[#2\right]}}

\newcommand{\affectedscript}[0]{{\mathcal{A}}}
\newcommand{\notaffectedscript}[0]{{\mathcal{N}}}

\acrodef{IR}{information retrieval}
\acrodef{EM}{Expectation-Maximization}
\acrodef{rbEM}{regression-based EM}
\acrodef{LTR}{learning to rank}
\acrodef{SERP}{Search Engine Results Page}
\acrodef{CLTR}{counterfactual learning to rank}
\acrodef{OLTR}{Online Learning to Rank}
\acrodef{IPS}{inverse propensity scoring}
\acrodef{PBM}{Position-Based Model}
\acrodef{CM}{Cascade Model}
\acrodef{ARP}{average relevance position}
\acrodef{DCG}{discounted cumulative gain}
\acrodef{NDCG}{normalized discounted cumulative gain}
\acrodef{DLA}{Dual Learning Algorithm}
\acrodef{MNAR}{Missing-Not-At-Random}
\acrodef{mPMM}{modified Poisson Mixture Model}
\acrodef{DCM}{Dependent Click Model}
\acrodef{UBM}{User Browsing Model}
\acrodef{CCM}{Click Chain Model}
\acrodef{PMM}{Poisson Mixture Model}
\acrodef{GMM}{Gaussian Mixture Model}
\acrodef{BMM}{Binomial Mixture Model}
\acrodef{CTR}{click-through rate}
\acrodef{AC}{affine correction}
\acrodef{i.i.d.}{independent and identically distributed}
\acrodef{CLT}{Central Limit Theorem}
\acrodef{MAP}{Mean Average Precision}
\acrodef{LOO}{leave one out}
\acrodef{SGD}{stochastic gradient descent}
\acrodef{GD}{gradient descent}
\acrodef{CV}{cross validation}
\acrodef{NN}{neural network}
\acrodef{DTR}{disparate treatment ratio}
\acrodef{EEL}{expected exposure loss}
\acrodef{KS}{Kolmogorov-Smirnov}

\theoremstyle{definition}

\DeclareMathOperator*{\argmin}{argmin}

\theoremstyle{definition}

\newtheorem{remark}{Remark}[]

\allowdisplaybreaks

\makeatletter
\renewcommand\paragraph{\@startsection{paragraph}{4}{\parindent}%
  {-.1\baselineskip \@plus -1\p@ \@minus -.1\p@}%
  {-2\p@}%
  {\ACM@NRadjust{\@parfont\@adddotafter}}}
\makeatother  

\newcommand{\ROne}[1]{\textcolor{teal}{\textbf{[R1]}}}
\newcommand{\RTwo}[1]{\textcolor{blue}{\textbf{[R2]}}}
\newcommand{\RThree}[1]{\textcolor{cyan}{\textbf{[R3]}}}

\makeatletter
\newcounter{tmp@cnt}
\newcommand*\@labelpunc{.}
\newcommand*\combine[1][2]{%
    \refstepcounter{enumi}
    \setcounter{tmp@cnt}{\value{enumi}}
    \addtocounter{enumi}{#1-1}
    \item[\thetmp@cnt--\theenumi\@labelpunc]}
\newcommand*\labeltype[2][]{\gdef\@labelpunc{#1}\renewcommand\thetmp@cnt{#2{tmp@cnt}}}

%% file: preamble/meta.tex
\setcopyright{rightsretained}
\copyrightyear{2024}
\acmYear{2024}
\acmConference[SIGIR '24]{Proceedings of the 47th International ACM SIGIR Conference on Research and Development in Information Retrieval}{July 14--18, 2024}{Washington D.C., USA}
\acmBooktitle{Proceedings of the 47th International ACM SIGIR Conference on Research and Development in Information Retrieval (SIGIR '24), July 14--18, 2024, Washington D.C., USA}
\acmDOI{}
\acmISBN{}

\ccsdesc[500]{Information systems~Learning to rank}

\keywords{Ranking fairness, Unbiased learning to rank}

%% file: preamble/authors.tex
\author{Ali Vardasbi}
\orcid{0000-0002-9342-5272}
\affiliation{%
  \institution{Spotify}
  \city{Amsterdam}
  \country{The Netherlands}
}
\email{aliv@spotify.com}

\author{Maarten de Rijke}
\orcid{0000-0002-1086-0202}
\affiliation{%
  \institution{University of Amsterdam}
  \city{Amsterdam}
  \country{The Netherlands}
}
\email{m.derijke@uva.nl}

\author{Fernando Diaz}
\orcid{0000-0003-2345-1288}
\affiliation{%
  \institution{Carnegie Mellon University}
  \city{Pittsburgh}
  \country{The USA}
}
\email{diazf@acm.org}

\author{Mostafa Dehghani}
\orcid{0000-0002-9772-1095}
\affiliation{%
  \institution{Google DeepMind}
  \city{San Francisco}
  \country{The USA}
}
\email{dehghani@google.com}

%% file: sections/01-introduction.tex
% !TEX root = ../main.tex

\vspace*{-1mm}
\section{Introduction}
\label{sec:intro}

\vspace*{-.25mm}
Online search and recommender systems leverage user interaction data to enhance their ranking quality.
When using human interactions, however, we need to account for human bias and the possibility of learning unfair ranking policies.
In the context of \ac{LTR}, the term \emph{bias} usually refers to unequal treatment of items with equal utility by users~\citep{joachims2005accurately}. % cases where users treat items with the same level of utility, differently.
% For example, position bias occurs when items at the top of a ranked list receive more clicks than those relevant lower down: higher items in a list absorb more \emph{exposure}.
Studies show that if the bias is ignored, this leads to a degradation in the ranking quality of a system trained on the user interactions~\citep{joachims2017unbiased, wang2018position, ai2018unbiased, vardasbi2020when}.
Correcting for bias is a necessary step for high-quality rankings, but it is not sufficient.
A system should return rankings that strive, to a certain extent, for \emph{fairness of exposure}. There are different definitions for fairness of exposure in ranking, leading to different metrics~\citep{raj2022measuring, singh-2018-fairness, diaz2020evaluating, wu2022joint, JIN2023101906, deldjoo2023fairness}. However, the core idea is the same: items with similar levels of utility should receive similar exposure by the system.
Without meeting fairness of exposure,  bias towards the privileged groups or individuals is reinforced, both in what the system learns from the ongoing interactions~\citep{fabris2020gender, hajian2016algorithmic, saxena2021exploring}, and in users' judgments about the utility of items~\citep{kay2015unequal,suhr2021does}.\looseness=-1

\input{figures_tex/example}

\vspace{0.2em}
\myparagraph{Group bias.}
A search or recommender system can only ensure that items with similar utility receive comparable exposure to users, by arranging them accordingly. 
However, this alone is insufficient.
Users' judgments about the utility of items are affected by their perception of the item's group membership~\citep{kay2015unequal, suhr2021does, krieg2022perceived, liu2023reducing}.
This means that, even when the exposure of two high-utility items from two different groups is the same, users may judge them differently and one group may receive more clicks than the other.
% \ali{
\cref{fig:example} provides a toy example for this phenomenon. 
Assume a job application system where there are four applicants coming from two groups, shown in the figure by squares ($S1$ and $S2$) and circles ($C1$ and $C2$).
The average relevance over the groups is equal and both groups receive almost equal exposure on average, i.e., the ranking is fair w.r.t. \ac{DTR} definition~\citep{singh-2018-fairness}.
As a result of equal average relevance and exposure, we expect both groups to receive an (almost) equal number of clicks from the employers.
This is shown by the equal-length green bars in the ``group average'' part of the figure.
However, employers' judgments are affected by group bias: they \emph{think} that candidates coming from the square group are more suitable for this type of job, compared to the circle group candidates.
Consequently, while being equally relevant and equally exposed, candidates from the circle group are selected less often, as shown by the blue bars in the ``group average'' part of the figure.
% }
We refer to this behavior as \emph{group membership bias}.
Our study focuses on the scenario of two groups, where the term ``affected'' refers to the group whose items are prone to underestimation and receive fewer clicks than they ideally should (the circle group in the example).
% Group bias is closely related to the concept of implicit bias, defined as unconscious and unintentional preference of individuals based on their membership in particular groups~\citep{kleinberg2018selection}.
% Implicit bias is a special case of group bias, as group bias includes intentional biases towards particular groups as well.
% 
% \myparagraph{Theoretical and experimental analysis.}
Considering clicks as the primary measure of user interaction, we provide both theoretical and empirical analyses for the impact of group biased clicks on ranking quality as well as two merit-based fairness of exposure metrics.\looseness=-1

% which are used in two ways:
% \begin{enumerate*}[label=(\roman*)]
%     \item For head queries, clicks are memorized and the ranking is performed via tabular search. In this case, the training labels are directly used to generate the ranking presented to the user, and user preferences can often be obtained with high accuracy~\citep{oosterhuis-2021-robust}.
%     \item For tail queries the clicks are used as supervision signals to train an \ac{LTR} model. In this case, the outputs of the \ac{LTR} model are used to generate the ranking that is shown to the user.
% \end{enumerate*}
% With that in mind, 
% We first provide a theoretical analysis of the impact of group bias on various metrics.
% This gives us an understanding of the impact of group bias on the head queries directly, and on the tail queries indirectly.
% Theoretically analyzing the output of a \ac{LTR} model involves considering the architecture and loss function of the LTR model, which is beyond the scope of this paper.
% For the experimental part, however, we analyze the impact of group bias both on the training labels as well as the outputs of an LTR model.

% \vspace{0.1em}
\myparagraph{Impact on ranking.}
% Previous work on implicit bias~\citep{kleinberg2018selection, Celis2020Interventions} has shown that similar to other types of bias, implicit bias can degrade the ranking quality of systems.
% In this work, we add to their theoretical results by quantifying this degradation with an approximation formula for the \ac{NDCG} metric.
Similar to other types of bias, group bias can potentially degrade the ranking quality of systems.
% \ali{
For example, in \cref{fig:example}, if the clicks are used to infer the relevance of the candidates, without correcting for group bias, inference would be biased towards members from the square group, negatively impacting the ranking quality of the system.
% }
In this work, we first theoretically quantify this degradation with an approximation formula for the \ac{NDCG} metric.
% Furthermore, we provide an important part that is missing in previous work, i.e., an experimental analysis of the impact of group bias on the outputs of an LTR model.
Then, we experimentally analyze the change in the ranking quality of an LTR model trained on clicks that suffer from group bias, compared to the full information case.
\cref{fig:introndcg} (left) shows an example of the impact of group bias on the ranking performance, measured by \ac{NDCG} on the Yahoo! dataset with feature number $426$ as the sensitive attribute (see~\cref{sec:setup}).
In this plot, the bars associated with the ``(non-) affected group'' label show the NDCG@10 when \emph{only} the relevant items from the (non-) affected group are considered relevant.
Note that a lower group underestimation factor means a higher group bias, and a factor equal to $1$ (the leftmost bars) means no group bias.
% On the left, we observe the impact on the training labels, i.e., tabular search, while the right plot shows the impact on the outputs of a general LTR model.
We observe that the affected group is hurt by group bias, while the other group has gained.
Importantly, the overall ranking quality is degraded by increasing group bias.\looseness=-1

% \vspace{0.1em}
\myparagraph{Impact on fairness.}
Unlike other types of bias that may affect fairness indirectly, group bias has a direct impact on fairness:
Clicks suffering from group bias can lead the system to undervalue the utility scores of a particular group (see, e.g., \cref{fig:introndcg}).
Consequently, when the expected exposure is assigned to groups based on these biased estimates of the utility, the ranking may not be \emph{truly} fair.
% \ali{
For example, in \cref{fig:example}, without correcting for group bias, the square members are inferred to be noticeably more relevant than the circle members and a fair system would swap $S2$ and $C2$ to give more exposure to the group that is more relevant.
This means that, based on the \emph{observed} clicks, the $[S1, C1, S2, C2]$ ranking is considered fairer than the original ranking in the figure.
But based on the \emph{latent} true relevance values, this ranking is far from being fair, giving the square group too much exposure.
% }
For our analyses, we consider two widely used metrics for fairness of exposure, namely \acfi{DTR}~\citep{singh-2018-fairness} and \acfi{EEL}~\citep{diaz2020evaluating, biega-2020-overview}.
Each metric has a definition for the \emph{ideal expected exposure} in terms of the utility, that leads to the fairest ranking.
Distinguishing between the true (unbiased) utility and observed (biased) utility, we provide formulas for the change in the true fairness metrics, when the target expected exposure is obtained from the biased utility.
% We also provide a thorough experimental analysis of this change, both on the training labels and the outputs of an LTR model.
\cref{fig:introndcg} (right) shows an example of the impact of group bias on the \ac{DTR} fairness metric.
% \ac{DTR} is a multiplicative metric, so we measure the ratio between the target exposure of the affected group, to the target exposure of the non-affected group.
% \footnote{We follow~\citep{singh-2018-fairness} and always keep the ratio below unity.
% , i.e., keep the less exposed group on the numerator.
% }
With \ac{DTR}, a ratio of $1$ means the fairest exposure, i.e., the leftmost bar with no group bias.
% We then normalize this ratio with the ratio obtained from the full information case, i.e., using true utilities without group bias; a \ac{DTR} score of $1$ means the fairest exposure.
% \ac{EEL} (right plot) is an additive metric, so we measure the $\ell_2$ loss between the target exposure vector in the biased case and the full information case.
% For \ac{EEL}, a loss of $0$ means the fairest exposure.
Similar to the ranking quality, here we also observe that group bias leads to noticeable deviations from the full information case in \ac{DTR} metric.

% \input{figures_tex/intro_fairness}

% \myparagraph{Impacts on users.}
% Shifting perception~\citep{kay2015unequal, vlasceanu2022propagation} (``shifting the representation of gender in image search results can shift people’s perceptions about real-world distributions.''): Is ``fair'' enough: Can a system change users' group bias by showing them fair rankings?
% Dis-satisfaction: Hypothesis: If the distance between the user's perception and the shown ranking is too much, the user may be dissatisfied. We will need to define a measure for user perception.

\input{figures_tex/intro_ndcg}

% \vspace{0.5em}
\myparagraph{Correction.}
% To correct for group bias, it should first be modeled. 
We follow previous work on implicit bias~\citep{kleinberg2018selection} and model group bias with a multiplicative factor.
This allows us to use \acfi{IPS} to correct for bias~\citep{joachims2017unbiased, wang2018position}.
Measuring group bias, however, is not as simple as measuring position bias.
We argue that group bias measurement requires assumptions on the distribution of the true utility scores. % of different groups.
Following~\citep{kleinberg2018selection, Celis2020Interventions, Emelianov2020On}, one can assume that the true utility scores of both groups come from the same distribution.
% Naively assuming that the utility scores for each query should come from the same distribution, leads to a trivial solution with equality instead of equity.
However, since \emph{equity} (i.e., merit-based fairness) is based on the premise that exposure should be distributed based on utility, assuming that the utility of different groups is equal for each query, means that different groups should receive equal exposure all the time, which means equality.
% Inspired by the notion of amortized fairness of exposure~\citep{biega2018equity}, 
To counter this equality-instead-of-equity issue, we propose to consider a set of queries (instead of one query) with their corresponding associated items and measure the group bias parameter over this aggregated set of scores.
We show that our correction method based on the above amortized measurement of the bias parameter is effective for restoring both the ranking quality and fairness metrics.%\looseness=-1

% \vspace{0.5em}
\myparagraph{Research questions.}
We answer the following questions:
\begin{enumerate}[label=(RQ\arabic*),leftmargin=*]
    \item What is the impact of group bias on the ranking quality and the true fairness metric of head and tail queries?
    \item How can we effectively correct for group bias, without substituting equality for equity?
\end{enumerate}

\noindent%
% \ali{
The rest of the paper is organized as follows. We first list several existing user studies showing that group bias exists in user interactions (\Cref{sec:relatedwork}).
In \Cref{sec:groupbias} we give a formal definition of group bias.
Answering the research questions starts from \Cref{sec:theory} where we derive mathematical formulas for the impact of group bias on the ranking quality and fairness metrics.
To do so, we put some simplifying assumptions such as binary latent relevance and uniform distribution of the observed clicks (i.e., attractiveness) over the items.
These assumptions may not be precisely met in real-world datasets, but we believe the derived formulas that are based on these assumptions are beneficial in giving insights into the impact of group bias on ranking.
To close the gap between theory and practice, we perform extensive experiments on various real-world datasets and show the negative impact of group bias on the ranking quality and fairness metrics in \Cref{sec:experiments:impact}.
Regarding the second research question, in \Cref{sec:correction} we discuss the challenges and solutions for measuring and correcting for group bias from user interactions.
Then we experimentally examine the effectiveness of our theoretical propositions for group bias correction in \Cref{sec:experiment}.
Finally, \Cref{sec:conclusion} concludes the paper.
% }

%% file: figures_tex/example.tex
\begin{figure}[t]
    \centering
    {
    \begin{tabular}[]{c@{\hspace{4em}}c}
        \multirow{4}{*}{
        \includegraphics[scale=1]{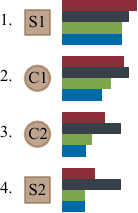}
        }
        & Group average
        \\
        &
        \includegraphics[scale=0.9]{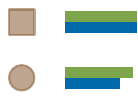}
        \\
        % &
        % \\
        &
        % \raisebox{.1\height}{
        \includegraphics[scale=0.5]{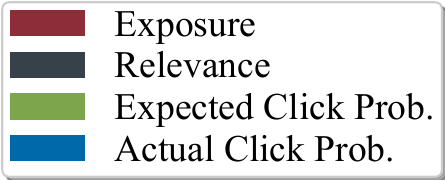}
        % }
        
    \end{tabular}
    }
    \caption{The effect of group bias on user clicks.}
    \label{fig:example}
\end{figure}

%% file: figures_tex/intro_ndcg.tex
% \begin{figure}[t]
%     \centering
%     {
%     \def\arraystretch{.5}
%     \begin{tabular}[]{@{}l@{\hspace{0.2em}}c@{\hspace{0.1em}}c@{}}
%         &
%         \multicolumn{2}{c}{
%         \hspace{0.1em}
%         \includegraphics[scale=0.38]{figures/intro_ndcg_legend}  
%         } \\
%         & Tabular (training labels)
%         & LTR output \\
%         \rotatebox[origin=lt]{90}{\hspace{1.2em} \small NDCG@10}
%         &
%         \includegraphics[width=0.22\textwidth]{figures/intro_tabular_ndcg}
%         &
%         \includegraphics[width=0.22\textwidth]{figures/intro_ltr_ndcg} \\
%         & \small Group underestimation factor
%         & \small Group underestimation factor
        
%     \end{tabular}
%     }
%     \caption{The impact of group bias on ranking performance for the Yahoo! dataset.}
%     \label{fig:introndcg}
% \end{figure}

\begin{figure}[t]
    \centering
    {
    \def\arraystretch{.5}
    \begin{tabular}[]{@{}l@{\hspace{0.15em}}c@{\hspace{0.8em}}l@{\hspace{0.15em}}c@{}}
        &
        \multicolumn{3}{c}{
        \hspace{0.1em}
        \includegraphics[scale=0.38]{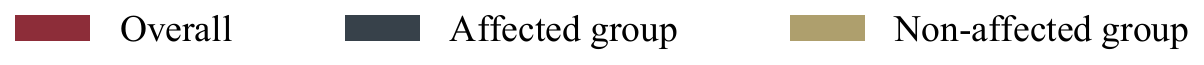}  
        } \\
        & \hspace{0.6em} Ranking quality $\uparrow$ (ideal: $1$)
        &
        & \hspace{0.6em} DTR $\uparrow$ (ideal: $1$)
        \\
        \rotatebox[origin=lt]{90}{\hspace{1.2em} \small NDCG@10}
        &
        \includegraphics[width=0.22\textwidth]{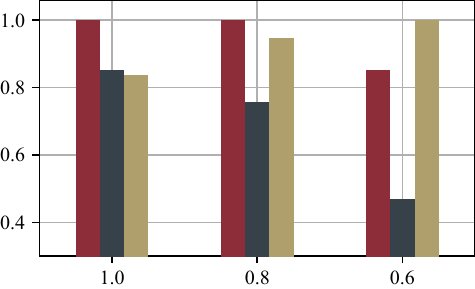}
        & \rotatebox[origin=lt]{90}{\hspace{0.25em} \small Exposure ratio}
        &
        \includegraphics[width=0.21\textwidth]{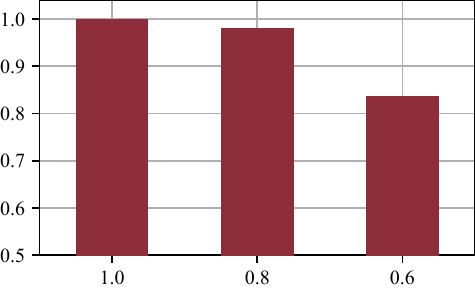}
        \\
        & \small Group underestimation factor
        &
        & \small Group underestimation factor
        
    \end{tabular}
    }
    \caption{The impact of group bias on ranking performance (left) and DTR fairness metric (right) for the Yahoo! dataset.}
    \label{fig:introndcg}
\end{figure}

%% file: sections/02-related.tex
% !TEX root = ../main.tex

\section{Related Work}
\label{sec:relatedwork}
An increasing number of studies indicate the existence of group bias.
Implicit bias, a special case of group bias, in which the preference of one group over the other is unintentional, has been widely studied in human behavior studies \citep[e.g.,][]{jolls2006law, brownstein2015implicit, fitzgerald2017implicit}.
More recently, implicit bias has been formalized in the set selection problem~\citep{kleinberg2018selection} and extended to the ranking scenario~\citep{Celis2020Interventions, Emelianov2020On}.

Here, we list a small number of example studies indicating that group membership affects users' judgment.
% Studies by \citet{lyness2006fit} on performance evaluations and promotions of managers indicate that standards for women promotion were stricter, e.g., women had to show roughly twice as much evidence of competence as men to be seen as equally competent.
In \citep{kay2015unequal}, it is observed in a user study that in a career search, results that are consistent with stereotypes for a career are rated higher.
\citet{suhr2021does} pose the important question of whether ``fair ranking improve[s] minority outcomes?'' and arrive at the result that persistent gender preferences of employers can limit the effectiveness of fair ranking algorithms.
% , and that fair ranking is more effective when the features of an underrepresented candidate are similar to the average overrepresented group features.
\citet{krieg2022perceived} in their user study on gender sensitive queries from~\citep{krieg2022grep} show that perceived gender bias affects judgment.
In \citep{vlasceanu2022propagation} it is shown that societal and algorithmic gender bias affect each other: the algorithmic outputs of search engines track pre-existing societal-level gender biases; and, at the same time, exposure of users to these biased results shape users’ cognitive concepts and decisions.
\citet{liu2023reducing} study gender bias in the evaluation and selection of future leaders.\looseness=-1

% All of the above examples confirm that group bias exists in user interactions.
We study the impact of group bias on ranking and fairness measures and propose a method to correct for it.
Closest to our paper are~\citep{Celis2020Interventions, mehrotra2022selection}, which show that implicit bias degrades ranking quality and that by ensuring equality of exposure, the ranking quality can be improved.
What we add on top of this work is to provide a formalization of the change of ranking and merit-based fairness metrics as a result of group bias.
We also provide experimental analyses of the impact of group bias on the output of an \ac{LTR} model.

% Furthermore, we propose an amortized correction method that does not replace equity with equality (unlike existing work on implicit bias).
The idea of our amortized correction to counter sparsity and equality-instead-of-equity has similarities to the notion of amortized fairness of exposure~\citep{biega2018equity}, where the exposure and utility of individuals (or groups) are aggregated across multiple queries and the fairness metric is calculated according to the aggregated exposure and utility.
This corresponds to fairness evaluation.
In contrast, we aggregate the items associated with multiple queries to find the group bias parameter that minimizes the distance between the utility distribution of the affected and non-affected groups.
This corresponds to group bias correction.

Recently, \citet{Balagopalan2023therole} have provided a set of desired criteria that relevance scores should satisfy in order to meaningfully guide fairness interventions, and show that, e.g., learning from click data can violate some of the desiderata of relevance for fair ranking.
Though their focus is not on group bias, their approach to showing that a violation of assumptions on the utility of items can misguide fairness interventions has similarities to this work.

\begin{remark}
    Our terminology of \emph{group bias} should not be confused with \emph{in-group bias}, where a user favors members from their own group over out-of-group members~\citep{Zdaniuk2001group, Molenberghs2013the, o2023racial}, or \emph{conformity bias}, where users tend to behave similarly to the others in a group~\citep{chen2023bias, karomat2023unconscious}.
    % In this study, we follow a body of previous work on implicit and explicit bias and only focus on the group membership of the items and do not consider the group of the users.
    Hence, issues such as loyalty versus neutrality are out of scope.
\end{remark}

%% file: sections/03-groupbias.tex
% !TEX root = ../main.tex

\section{Group Membership Bias}
\label{sec:groupbias}

% \subsection{Definitions}
As discussed in \cref{sec:relatedwork}, prior work shows that the judgment of a user about the relevance of an item may be affected by the item's group.
Either unconsciously (as in implicit bias~\citep{kleinberg2018selection, Emelianov2020On}) or due to stereotypical bias~\citep{kay2015unequal, saxena2021exploring, krieg2022perceived}, users tend to rate one group higher than the other.
In this paper, we do not aim to deal with the source of this biased behavior and only focus on its impact on algorithms and metrics.
We call this behavior the \emph{group membership bias}.
Following the well-known examination hypothesis~\citep{chuklin2015click} that says that an item is clicked by a user if it is
\begin{enumerate*}[label=(\roman*)]
    \item examined and
    \item found attractive
\end{enumerate*}
by that user, one can attribute group bias to the attractiveness part:
\begin{equation}
\label{eq:attraction}
    P(A \mid q, d, g) = f\big(P(R \mid q, d), g\big),
\end{equation}
where $A$, $R$, $q$, and $d$ stand for attractiveness, relevance, query, and document, respectively, and $g$ is the group of which $d$ is a member.
Eq.~\eqref{eq:attraction} states that the attraction of an item to the user not only depends on the item's true relevance to the query, but is also a function of the item's group.
Following the literature on implicit bias and gender bias~\citep{kleinberg2018selection, saxena2021exploring}, we assume this dependency to have a multiplicative form as follows:\looseness=-1
\begin{equation}
\label{eq:implicitbias}
    P(A \mid q, d, g) = \beta_{g} \cdot P(R \mid q, d).
\end{equation}
We call $\beta_{g}$ the group underestimation factor, or \emph{group propensity}.
% \ali{
Our implicit assumption is that clicks for the affected group are missing completely at random (MCAR) with $\beta_g$ being the missingness probability. This brings the bias-correction problem back to \ac{IPS} correction, since the clicks missing due to group bias are analogous to clicks missing based on position.
% }
Notice that $\beta_g$ is not necessarily fixed across all queries.
For instance, in the Grep-BiasIR dataset \citep{krieg2022grep}, bias-sensitive queries have different expected gender stereotypes, and users are expected to be biased toward the respective gender stereotype of each query.

\begin{remark}
For simplicity of notation, we consider one sensitive attribute here.
Extending our discussions to more attributes with intersectional groups is possible using the formulation in~ \citep{Celis2020Interventions,mehrotra2022selection}.
\end{remark}

\begin{remark}
    % \ali{
    Group bias is only meaningful in settings where there is a global (objective) notion of relevance in contrast to personalized (subjective) relevance.
    For example, when ranking students for college entrance, or selecting among job applicants hiring, it is desired to base the decision on an \emph{unbiased} and \emph{unpersonalized} criterion.
    On the other hand, when searching for a roommate, the relevance is subjective and the group bias concern does not apply.
    % }
\end{remark}

\subsection{Ranking Regimes}
We distinguish between two \ac{LTR} regimes:
\begin{enumerate*}[label=(\roman*)]
    \item tabular search for head queries; and
    \item general LTR model for other queries.
\end{enumerate*}
Note that the majority of previous studies focused only on the general LTR regime \citep[e.g.,][]{morik2020Controlling, singh2019policy}, or the tabular regime \citep[e.g.,][]{singh-2018-fairness, biega2018equity}.
In contrast, we follow~\citep{vardasbi2022probabilistic} and consider both LTR regimes.

\myparagraph{Tabular search for head queries.}
In tabular search, users' historical interactions with head queries are directly used to estimate items' utility~\citep{katariya2016dcm,kveton2015cascading,lagree2016multiple,lattimore2019bandit,li2020bubblerank,radlinski2008learning,zoghi2016click}.
In this regime, we assume that $P(A)$ can accurately be inferred from clicks: other types of bias such as position and trust bias are corrected for and only group bias remains in the signals.
% In addition to their direct usage in tabular search, these accurate attractiveness estimations are used as training signals for the general LTR model (see next paragraph) as well.
Our theoretical results on the impact of group bias on different metrics, 
% i.e., the effect of using the observable attractiveness, $P(A)$, instead of the hidden true relevance, $P(R)$, for optimizing the ranking and fairness metrics, 
lie in this regime.\looseness=-1 
% Then, in RQ2, we deal with effectively correcting for group bias, i.e., inferring $P(R)$ from $P(A)$, with some assumptions on the utility score distributions.

\myparagraph{General LTR model for tail queries.}
% Not all queries are abundantly available in historical interaction logs.
For new queries and ones that a tabular model is not confident about, an LTR model is used.
We assume that this LTR model is trained over accurate estimations of $P(A)$ from the head queries.
Writing $r_{q,d}$ for the relevance of item $d$ to query $q$, ranking metrics per query can usually be expressed as: $\Delta_q = \textstyle \sum_{d} \lambda_{q,d} \cdot r_{q,d}$,
% \begin{equation}
%     \label{eq:ltr}
%     \Delta_q = \textstyle \sum_{d} \lambda_{q,d} \cdot r_{q,d},
% \end{equation}
where $\lambda$ is a metric-specific coefficient.
% e.g., for \ac{DCG} we have $\lambda^{DCG}_{q,d}=-\log\big(1+\text{rank}(d)\big)^{-1}$.
Using attractiveness instead of the true relevance to train an LTR model, means that instead of $\Delta_q$, the following metric is being optimized:
\begin{equation}
    \label{eq:biasedltr}
    \hat{\Delta}_q = \sum_{d} \lambda_{q,d} \cdot P(A \mid q,d) \stackrel{\eqref{eq:implicitbias}}{=} \sum_{g} \beta_g \sum_{d\in g} \lambda_{q,d} \cdot P(R \mid q,d).
\end{equation}
Comparing $\Delta_q$ and \cref{eq:biasedltr}, it is easy to see that $\hat{\Delta}_q$ is biased:
\begin{equation}
    \mExpect{R}{\Delta_q} = \textstyle \sum_{d} \lambda_{q,d} \cdot P(R \mid q,d) \neq \hat{\Delta}_q,
\end{equation}
unless $\beta_g$ is equal across all groups, i.e., there is no group bias.
% Similar to studies on position and trust bias~\citep{joachims2017unbiased, agarwal2019addressing, vardasbi2020cascade, vardasbi2020when, vardasbi2021mixture}, in our experiments, we analyze the effect of group bias on the ranking quality of the LTR model (RQ1).
% Due to the relationship between group bias and fairness concerns, we go one step further and assess how leaving the group bias uncorrected affects the optimization of fairness metrics.
% Regarding our second research question, we show the effectiveness of our amortized correction method to improve the ranking quality and fairness metrics on the output scores of the LTR model.

%% file: sections/04-theory.tex
% !TEX root = ../main.tex

\section{Theoretical Results}
\label{sec:theory} 
We assume that there are two groups $G_{\affectedscript}$ (affected) and $G_{\notaffectedscript}$ (non-affected), with $\beta_{\affectedscript} < 1$ and $\beta_{\notaffectedscript}=1$.
We further assume binary latent relevance, and that within each group, relevant items are more attractive than non-relevant items:
% \footnote{This assumption is consistent with the fact that tabular search for head queries usually leads to highly accurate relevance estimations~\citep{sorokina2016amazon, oosterhuis-2021-robust}.}
\begin{equation}
    \label{eq:monotoneattractiveness}
    \mbox{}\hspace*{-2mm}
    \forall d,d' \!\in\! G_i, \text{ if } 
    r_d=1 \text{ and } r_{d'}=0, \,\text{then}
    ~ P(A \mid q, d) \!>\! P(A \mid q, d').
    \hspace*{-2mm}\mbox{}
\end{equation}
For brevity, we write $a_d$ for the attractiveness probability of item $d$, assuming that there is no confusion about the query.
Let the number of candidate items for a query be $n$, out of which $n_{\affectedscript}$ and $n_{\notaffectedscript}$ items belong to groups $G_{\affectedscript}$ and $G_{\notaffectedscript}$, respectively.
We indicate the number of relevant items with $n^+_{\affectedscript}$ and $n^+_{\notaffectedscript}$.
To assess the impact of group bias on different metrics, we measure the change in the target metric when the observable attractiveness probabilities are considered as a proxy for the true relevance scores.

\subsection{Ranking Quality}
For  ranking quality, we calculate the \ac{NDCG} of the list obtained from sorting items based on their attractiveness probability and measure its deviation from the ideal \ac{NDCG}, i.e., $1$.
By definition, group bias affects the attractiveness probabilities for $G_{\affectedscript}$ only.
Let $a^*$ be the minimum attractiveness value for the relevant items of $G_{\notaffectedscript}$:\looseness=-1
\begin{equation}
\label{eq:astar}
    a^* = \min_{d\in G_{\notaffectedscript}} \{a_d \mid r_d = 1\}.
\end{equation}
Items in $G_{\affectedscript}$ with higher attractiveness values than $a^*$ are ranked correctly with probability $1$: Group bias has dampened their attractiveness probabilities, but still none of the non-relevant items is ranked higher than them.
We define an auxiliary random variable $\nu$ to be the fraction of relevant items from the affected group $G_{\affectedscript}$ that are ranked correctly with probability $1$:
\begin{equation}
    \label{eq:nu}
    \nu = \tfrac
    { |\{d \mid d \in G_{\affectedscript} \land r_d = 1 \land a_d > a^*\}| }
    { |\{d \mid d \in G_{\affectedscript} \land r_d = 1\}| }.
\end{equation}
For uniformly distributed scores in the interval of $[0,1]$, we have $\mExpect{}{\nu}=\max ( 2-{\beta_{\affectedscript}^{-1}}, 0 )$.
% the mean value of $\nu$ has a closed form as follows:
% \begin{equation}
%     \label{eq:meannu}
%     \mExpect{}{\nu}=
%     \begin{cases}
%          2-\frac{1}{\beta_{\affectedscript}}, &  \text{if }\beta_{\affectedscript}>0.5\\\
%          0, & \text{otherwise.}
%     \end{cases}
% \end{equation}
% \begin{equation}
%     \label{eq:meannu}
%     \mExpect{}{\nu}=\max \left( 2-{\beta_{\affectedscript}^{-1}}, 0 \right).
% \end{equation}

\begin{theorem}
In the presence of group bias, for uniformly distributed attractiveness scores, the change in the \ac{NDCG} of the list, sorted based on items' attractiveness, can be approximated by a linear function of $\mExpect{}{\nu}$, i.e., the fraction of affected relevant items that are still as attractive as the non-affected relevant items.
\end{theorem}
\begin{proof}
Our monotonicity assumption of the within-group attractiveness (\cref{eq:monotoneattractiveness}) ensures that no relevant item is ranked lower than non-relevant items of $G_{\affectedscript}$.
This means that the $1-\nu$ fraction of the affected relevant items lies somewhere between $n^+_{\notaffectedscript}+\nu n^+_{\affectedscript}$ and $n_{\notaffectedscript} + n^+_{\affectedscript}$ positions.
The expected \ac{DCG} of the list would be as follows:
\begin{equation}
\label{eq:averagedcg}
    \mExpect{}{DCG}=\sum_{i=1}^{n^+_{\notaffectedscript}+\nu n^+_{\affectedscript}} 
            \frac{1}{\log(1+i)} +
        \sum_{i=n^+_{\notaffectedscript}+\nu n^+_{\affectedscript}+1}^{n_{\notaffectedscript} + n^+_{\affectedscript}}
            \frac{\xi_i}{\log(1+i)},
\end{equation}
where $\xi_i$ depends on the distribution of the attractiveness scores.
For a uniform distribution, we have:
\begin{equation}
    \xi_i = \tfrac
    {(1-\nu)n^+_{\affectedscript}}
    {n_{\notaffectedscript} -n^+_{\notaffectedscript}+(1-\nu) n^+_{\affectedscript}}.
\end{equation}
Finally, using numerical analysis to approximate the average \ac{DCG} in \cref{eq:averagedcg} by a linear function of $\nu$, leads to a small approximation error, e.g., a relative error of at most $5\%$ in a top-$20$ setup.
\end{proof}

\subsection{Merit-Based Fairness Metrics}
Next, to see the impact of group bias on fairness metrics, we analyze two well-known merit-based fairness of exposure metrics, viz.\ \ac{EEL}~\citep{diaz2020evaluating, biega-2020-overview} and \ac{DTR}~\citep{singh-2018-fairness}.
For both, we calculate the target exposure in two cases: 
\begin{enumerate*}[label=(\roman*)]
\item the full information case where the true relevance scores are used to compute target exposure, and 
\item the group biased case where attractiveness probabilities are used as proxies for relevance to compute target exposure.
\end{enumerate*}
By \emph{change in true target exposure} we mean the difference between these two cases.

\subsubsection{EEL}
In the next theorem, we calculate the change in the target exposure of $G_{\notaffectedscript}$ as a result of group bias.
\begin{theorem}
    In the presence of group bias, assuming the \ac{PBM} as the user browsing model with logarithmic decay of exposure as in \ac{DCG},\footnote{Here we follow~\citep{trec-fair-ranking-2021, vardasbi2022probabilistic}. Similar analyses can be performed for other exposure models.} the change in the target exposure of \ac{EEL} can be approximated as follows:
    \begin{equation}
    \Delta(\text{EEL}) = c \cdot \log\left(\tfrac{\text{\# True relevant items}}
    {\text{\# Perceived relevant items}}\right),
\end{equation}
where $c$ is a constant depending on $n^+_{\notaffectedscript}$, $n_{\notaffectedscript}$, and $n^+_{\notaffectedscript}$.
\end{theorem}
\begin{proof}
As we are working with two groups, and the sum of the group exposures is fixed, to measure the change in the target exposure vector, it is sufficient to measure the change in the target exposure of one group and multiply it by $2$.

To compute the expected exposure for EEL, the utility values should be discrete.
With a slight abuse of notation, we assume that $a^*$ (instead of \cref{eq:astar}) is the threshold used for discretization of the attractiveness probabilities,\footnote{Usually, $a^*=0.5$ is the least controversial threshold.} and we use $\Bar{a}_d$ for the discretized value of $a_d$.
Since $\beta_{\notaffectedscript}=1$, we assume that $\Bar{a}_d=r_d$ for $d\in G_{\notaffectedscript}$.
However, for the affected items, because $\beta_{\affectedscript} < 1$, not all the scores are necessarily correct.
We re-use $\nu$ from \cref{eq:nu} to show the fraction of affected relevant items that are still recognized as relevant.

For the average exposure of the relevant and non-relevant items we use the following two approximations:
\begin{align}
\label{eq:exposureapprox}
    \tfrac{1}{m}\textstyle \sum_{i=1}^{m}\tfrac{1}{\log(1+i)} & \approx \alpha \log(m) + c
\\
\label{eq:exposureapprox2}
    \tfrac{1}{n-m}\textstyle \sum_{i=m}^{n}\tfrac{1}{\log(1+i)} & \approx \alpha' \log(m) + c',
\end{align}
where $\alpha$ and $\alpha'$ are constants, obtained by numerical analysis.
For example, for $n=20$, $\alpha=-0.146$ and $\alpha'=-0.022$ lead to relative approximation errors of at most $5\%$.
In the full information case, there are a total of $n^+_{\notaffectedscript}+n^+_{\affectedscript}$ relevant items, i.e., $m=n^+_{\notaffectedscript}+n^+_{\affectedscript}$ in~\cref{eq:exposureapprox} and~\eqref{eq:exposureapprox2}.
% and $n_{\notaffectedscript}-n^+_{\notaffectedscript}+n_{\affectedscript}-n^+_{\affectedscript}$ non-relevant items.
% As such, each relevant and non-relevant item should ideally receive an average exposure of:
% \begin{equation}
%     E_{\text{rel}} = \frac{1}{n^+_{\notaffectedscript}+n^+_{\affectedscript}}
%     \sum_{i=1}^{n^+_{\notaffectedscript}+n^+_{\affectedscript}} \frac{1}{\log(1+i)},
% \end{equation}
% and
% \begin{equation}
%     E_{\text{non-rel}} = \frac{1}
%     {n_{\notaffectedscript}-n^+_{\notaffectedscript}+n_{\affectedscript}-n^+_{\affectedscript}}
%     \sum_{i=n^+_{\notaffectedscript}+n^+_{\affectedscript}+1}
%     ^{n_{\notaffectedscript}+n_{\affectedscript}} \frac{1}{\log(1+i)},
% \end{equation}
% respectively.
But with group bias, only $m=n^+_{\notaffectedscript}+\nu n^+_{\affectedscript}$ of the items are recognized as relevant.
% Consequently, the average exposure for each relevant and non-relevant item changes to:
% \begin{equation}
%     \tilde{E}_{\text{rel}} = \frac{1}{n^+_{\notaffectedscript}+\nu n^+_{\affectedscript}}
%     \sum_{i=1}^{n^+_{\notaffectedscript}+\nu n^+_{\affectedscript}} \frac{1}{\log(1+i)},
% \end{equation}
% and are respectively $n^+_{\notaffectedscript}$ and $n_{\notaffectedscript}-n^+_{\notaffectedscript}$.
Consequently, the change in the target exposure as a result of group bias can be approximated as follows:\looseness=-1
\begin{equation*}
\label{eq:deltaeel}
    \Delta(\text{EEL}) = 2\big(\alpha n^+_{\notaffectedscript} + \alpha'(n_{\notaffectedscript}-n^+_{\notaffectedscript})\big)
    \log\left(\tfrac{n^+_{\notaffectedscript}+n^+_{\affectedscript}}
    {n^+_{\notaffectedscript}+\nu n^+_{\affectedscript}}\right).
    \qedhere
\end{equation*}
\renewcommand{\qedsymbol}{}
\end{proof}

\subsubsection{DTR}
DTR is a multiplicative metric.
To have a meaningful measure for the change in DTR in the presence of group bias, one has to compute the ratio of the target expected exposure in the full information ($E$) and group-biased ($\Tilde{E}$) settings.
\begin{theorem}
In the presence of group bias, the change in the target exposure of \ac{DTR}, equals the fraction of affected relevant items that are still as attractive as the non-affected relevant items.
\end{theorem}
\begin{proof}
Using the same notation as in previous sections, and noting that because of the binary relevance assumption the utility of each group is equal to the number of its relevant items, this ratio is computed as follows:
\begin{equation*}
    \rho(\text{DTR}) = 
    \tfrac{\Tilde{E}_{\affectedscript}}{\Tilde{E}_{\notaffectedscript}} \cdot \tfrac{{E}_{\notaffectedscript}}{{E}_{\affectedscript}} =
    \tfrac{\nu n^+_{\notaffectedscript}+\nu n^+_{\affectedscript}}{n^+_{\notaffectedscript}+n^+_{\affectedscript}} =
    \nu.
    \qedhere
\end{equation*}
\renewcommand{\qedsymbol}{}
\end{proof}
% \begin{align}
%     \rho(E_{\notaffectedscript}^{\text{DTR}}) & = 
%     \frac{\Tilde{E}^{\text{DTR}}_{\notaffectedscript}}{E^{\text{DTR}}_{\notaffectedscript}} =
%     1 +
%     \frac{(1-\nu) n^+_{\affectedscript}}
%     {n^+_{\notaffectedscript}+\nu n^+_{\affectedscript}}
%     \\
%     \rho(E_{\affectedscript}^{\text{DTR}}) & = 
%     \frac{\Tilde{E}^{\text{DTR}}_{\affectedscript}}{E^{\text{DTR}}_{\affectedscript}} =
%     1 -
%     \frac{(1-\nu) n^+_{\notaffectedscript}}
%     {n^+_{\notaffectedscript}+\nu n^+_{\affectedscript}}.
% \end{align}

% \subsubsection{Comparison}
% In the extreme case \mdr{which is what}, it seems that DTR is more affected by the group bias, as when $\nu=0$, i.e., neither of the affected relevant items looks as attractive as not affected relevant items, the target expected exposure of the affected group is reduced to zero.
% \todo{...}

\myparagraph{Upshot.} 
% \ali{
In this section we derived formulisations of the negative effect of group bias on ranking. In proving our theorems, we relied on some assumptions on the relevance and distribution of attractiveness.
The objective of this section is to provide theoretical insights on the impact of group bias on ranking, i.e., RQ1.
Since our assumptions here may not be precisely met in real-world scenarios, in \Cref{sec:experiments:impact} we get back to RQ1 in a semi-synthesized experimental setup with real-world utility scores and show that, though not necessarily in complete alignment with the formulisations of this section, the message is the same: group bias \emph{does} hurt the quality and fairness of ranking in real-world scenarios, and the degradation gets stronger with more severe group bias.
% }

%% file: sections/05-correction.tex
% !TEX root = ../main.tex

\section{Group Bias Correction}
\label{sec:correction}
Our multiplicative formulation of group bias in \cref{eq:implicitbias} allows us to use \ac{IPS} to correct for group bias, once we know the value of the propensity $\beta$.
The unbiasedness proof of IPS for this case is exactly the same as that of position bias in~\citep{joachims2017unbiased, wang2018position}.
However, similar to position bias, the unbiasedness proof depends entirely on an accurate estimation of the bias parameter~\citep{vardasbi2021mixture}.

Unlike position bias, group bias cannot be measured by intervening in the ranked list of items.
The reason is that the bias attribute in position bias can be changed without modifying the content of the items: Each item can be shown in different positions, hence, detaching propensity from relevance.
In contrast, for group bias, the bias attribute, i.e., group membership, is a characteristic of the item that \emph{cannot} be changed.
As such, users' interactions with items cannot be measured for different values of the bias attribute.
Instead, to measure group bias, previous work on implicit bias (with the same problem formulation as \cref{eq:implicitbias}), assumes that the utility scores of different groups come from the same distribution~\citep{kleinberg2018selection, Celis2020Interventions, Emelianov2020On}.
We use the same assumption, but extend it to an amortized criterion.
% , inspired by the notion of amortized fairness~\citep{biega2018equity}.
% We aggregate all the scores across queries with the same level of expected group bias and find the $\beta$ such that the aggregated scores from different groups are most likely to come from the same distribution.
% Below, we first formalize group bias measurement, and then, noticing the issues of measuring a separate group propensity for each query, we propose our amortized group bias measurement method.

\subsection{Measurement}
\label{sec:correction:measurement}
Let $\mathbf{A}_{\affectedscript}$ and $\mathbf{A}_{\notaffectedscript}$ be the set of (observed) attractiveness scores, and $\mathbf{R}_{\affectedscript}$ and $\mathbf{R}_{\notaffectedscript}$ the set of (latent) relevance scores of $G_{\affectedscript}$ and $G_{\notaffectedscript}$, respectively.
Let $\Delta_{\mathbb{D}}$ be a non-parametric test for the equality of one-dimensional probability distributions such as the \ac{KS}~\citep{massey1951kolmogorov} test.
The assumption that the utility scores of the two groups come from the same distribution means that:
\begin{equation}
\label{eq:samedistributionassumption}
    \lim_{
        |\mathbf{R}_{\affectedscript}|, 
        |\mathbf{R}_{\notaffectedscript}| 
        \rightarrow \infty
        } 
    \Delta_{
        \mathbb{D}
        }
        \left(
            \mathbf{R}_{\affectedscript}, 
            \mathbf{R}_{\notaffectedscript}
        \right) = 0.
\end{equation}
Assuming \cref{eq:implicitbias} to be the relation between $\mathbf{A}_{\affectedscript}$ and $\mathbf{R}_{\affectedscript}$, the best estimation of $\beta_{\affectedscript}$ is given by the following optimization problem:
\begin{equation}
\label{eq:bestbeta}
    \hat{\beta}_{\affectedscript} =
    \argmin_{
        \beta_{\affectedscript}
        }
    \Delta_{
        \mathbb{D}
        }
        \left(
            \tfrac{\mathbf{A}_{\affectedscript}}{\beta_{\affectedscript}}, 
            \mathbf{A}_{\notaffectedscript}
        \right),
\end{equation}
where $\mathbf{A}_{\affectedscript}/\beta_{\affectedscript}$ is the set obtained by dividing all the scores in $\mathbf{A}_{\affectedscript}$ by $\beta_{\affectedscript}$.
% \mdr{Please check notation, as $\mathbf{A}/\beta$ is not actually used in the equation (as it stands).}
In our experiments, we choose the \ac{KS} test for $\Delta_{\mathbb{D}}$ and use grid search to solve the one-dimensional optimization of \cref{eq:bestbeta}.
\looseness=-1

It only remains to define how the sets $\mathbf{A}_{\affectedscript}$ and $\mathbf{A}_{\notaffectedscript}$ should be constructed.
Naively constructing these sets per query has two issues:
\begin{enumerate*}[label=(\roman*)]
    \item \emph{Sparsity}: Usually, we do not have a large number of items with non-zero exposure, associated with one query in real-world search engines. On the other hand, statistical tests measuring the distance between probability distributions work best with large numbers of data points. 
    % This means that considering the items of one query in \cref{eq:bestbeta} may not lead to reliable solutions due to high variance.
    \item \emph{Equality-instead-of-equity}: Assuming the same distribution for the utility of different groups can make the notion of \emph{equity} meaningless, as the implicit assumption in merit-based fairness metrics is that different groups may have different utilities. %Correcting the utility estimations in such a way that the utility scores of different groups are forced to have close distributions makes the target exposure of different groups almost equal.
    % equity will be replaced by equality.
\end{enumerate*}

% \vspace{-0.7em}
\begin{remark}
    Our assumption that the utility scores come from the same distribution derives from the principle of maximum entropy: unless there are explicit and justified reasons indicating that different groups have different utility score distributions, it is only reasonable to assume the same distribution.
    Prior work on implicit bias~\citep{kleinberg2018selection, Celis2020Interventions, Emelianov2020On} is based on this same assumption.
    % Moreover, if such a justified reason exists, Eq. \eqref{eq:bestbeta} can be changed accordingly.
\end{remark}
% Next, we try to address these two issues by using an amortized method.

\subsection{Amortized Correction}
\label{sec:correction:amortized}
Instead of using \cref{eq:bestbeta} per-query, in the amortized correction method we consider a set of queries with the same group propensity and aggregate the utility scores of their associated items. 
The sets $\mathbf{A}_{\affectedscript}$ and $\mathbf{A}_{\notaffectedscript}$ contain the attractiveness scores of these aggregated items.
This aggregation addresses both issues mentioned in \cref{sec:correction:measurement}:
\begin{enumerate*}[label=(\roman*)]
    \item With multiple queries, the size of the sets $\mathbf{A}_{\affectedscript}$ and $\mathbf{A}_{\notaffectedscript}$ grows, reducing the variance.
    \item Amortized equality does not force per-query equality. 
    % In other words, with amortized equality, the utility of different groups in each query can still be \emph{different} and distributing the exposure based on utilities, i.e., equity, does not necessarily lead to equality.
\end{enumerate*}
% \ali{
\cref{fig:exampleamortized} shows an example of the difference between per-query and amortized correction.
In this example, there are three queries, each with two results from the squares group and two from the circles group.
In all of the queries, the circles group is the affected group.
In per-query correction for group bias, the bias parameter is over-estimated for query $1$ (the blue bar is shorter than the relevance for the circles group members), but under-estimated for query $2$ (the blue bar is longer).
For query $3$, since the square members have a lower average number of clicks compared to the circle members, the per-query correction wrongly detected the squares group as the affected group and the discrimination between the groups has been boosted after the correction.
On the other hand, by aggregating the clicks on all six square members and six circle members and measuring a single group bias parameter for all the queries, the amortized correction has a noticeably better performance at recovering the true relevance (the green bars are close to the red bars).\looseness=-1
\input{figures_tex/example_amortized}

% }
% As an example, assume there are two queries with $\beta_{\affectedscript}=0.8$, each with two items from different groups.
% The relevance probabilities of the items from the $G_{\affectedscript}$ and $G_{\notaffectedscript}$ are $0.5$ and $1$ for the first query, and $1$ and $0.5$ for the second, respectively.
% In the per-query approach, two different propensities are inferred for the two queries to make the corrected utilities of the two groups equal in each query, leading to equality of exposure.
% However, in the amortized approach, we have $\mathbf{A}_{\affectedscript}=\{0.4, 0.8\}$ and $\mathbf{A}_{\notaffectedscript}=\{1, 0.5\}$.
% Solving \cref{eq:bestbeta} gives us a single value for $\beta$. In this case, groups will have different utility scores after correction, and equity of exposure has not been replaced by equality.\looseness=-1

The amortized correction, however, introduces a new challenge:
How to detect queries with the same group propensity, before measuring their group propensity?
One way to break this cyclic dependency is by using extra knowledge.
Notice that in order to detect queries with almost the same group propensity, it is only required to have a clustering of queries.
Previous work shows that such a clustering exists for a number of group attributes such as gender~\citep{krieg2022grep}.
In this paper, we first assume that such a clustering of queries is given.
Then, in an ablation experiment, we further show that even loosely clustering the queries, when an accurate and more specific clustering is not available, improves the ranking quality and fairness metrics over the naive case of not correcting for group bias.\looseness=-1

\myparagraph{Upshot.}
We relied on prior studies for the existence of group bias in user interactions and provided theoretical results about its impact on the ranking and merit-based fairness metrics.
Then, we proposed an amortized correction method for group bias that addresses the 
% sparsity as well as 
equality-instead-of-equity issue of the per-query correction.
Next, we test our theoretical findings and arguments experimentally.\looseness=-1

%% file: figures_tex/example_amortized.tex
\begin{figure}[t]
    \centering
    {
    \begin{tabular}[]{c@{\hspace{1.5em}}c@{\hspace{1.5em}}c}
        Query 1 & Query 2 & Query 3 \\
        \includegraphics[scale=0.9]{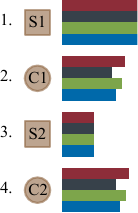}
        & 
        \includegraphics[scale=0.9]{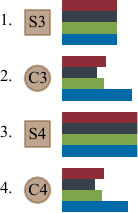}
        & 
        \includegraphics[scale=0.9]{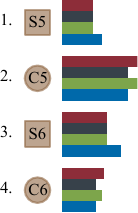}
        \\
        \multicolumn{3}{c}{
        \includegraphics[scale=0.42]{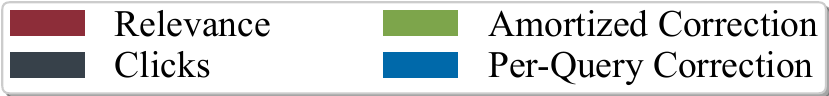}
        }
        
    \end{tabular}
    }
    \caption{Contrasting per-query and amortized correction.
    \vspace{-0.5em}
    }
    \label{fig:exampleamortized}
\end{figure}

%% file: sections/06-experiment.tex
% !TEX root = ../main.tex

\section{Experimental Results}
\label{sec:experiment} 
In our experiments we investigate the following questions regarding group bias:
\begin{enumerate*}[label=(\roman*)]
    \item Is the impact of group bias on degrading the ranking quality and fairness metrics consistent for different sensitive attributes and in different datasets?
    \item Can our correction method effectively correct for group bias?
    \item How does the amortized approach compare to the per-query approach for correction?
    \item How robust is our correction method to the accuracy of clustering the queries based on their group propensity?
\end{enumerate*}

% \vspace{-1em}
\subsection{Setup}
\label{sec:setup}

\myparagraph{Dataset.}
We use four datasets with provided sensitive attributes and two with synthesized sensitive attributes.
\begin{enumerate*}[label=(\roman*)]
    \item {\bfseries IIT-JEE}: The dataset comprises the scores of candidates who took the Indian Institutes of Technology Joint Entrance Exam (IIT-JEE) in 2009. This information was made public in June 2009, following a Right to Information request~\citep{kumar2009complaint}. It contains the scores of about $385k$ students, the student’s gender ($98k$ women and $287k$ men), their birth category (see~\citep{baswana2019centralized}), and zip code.
    This dataset was used in prior work on implicit bias~\citep[e.g.,][]{Celis2020Interventions}.
    We normalize the scores to the $[0,1]$ interval using min-max normalization.
    Furthermore, we simulate queries by grouping the students based on their birth category and zip code.
    This gives $48.6k$ queries, among which we only keep the ones with both genders and at least one normalized score above $0.5$ and one below $0.5$.
    The filtering gives us $2.9k$ queries with a total of $205k$ scores.
    \item[(ii--iii)] {\bfseries TREC 2019 and 2020}: The academic search dataset provided by the TREC Fair Ranking track 2019 and 2020~\citep{trec-fair-ranking-2019}. These datasets come with $632$ and $200$ training queries, respectively, with an average of $6.7$ and $23.5$ documents per query.
    Following~\citep{sarvi2021understanding, vardasbi2022probabilistic}, we divide the items (i.e., papers) into two groups based on their authors' h-index.
    \item[(iv)] {\bfseries MovieLens $1M$}: The classic movie recommendation dataset comprising $1M$ movie ratings that were provided by $6k$ users for $3.9k$ different movies. We scraped IMDB to obtain the country of origin and box office cumulative worldwide gross values for each item (i.e., movie). For the sensitive attributes, we consider two groupings as follows.
    In $\text{MovieLens}_{[Co.]}$, we divide the movies based on their first listed country of origin into United States (US) and non-US groups with $2.7k$ and $1.2k$ movies and $807k$ and $193k$ ratings, respectively.
    In $\text{MovieLens}_{[BO]}$, we divide the movies based on their box office with a threshold of $100M$\$ into high and low-grossing groups with $388$ and $3.5k$ movies and $324k$ and $676k$ ratings, respectively.
\end{enumerate*}

\myparagraph{LTR dataset.}
To analyze the impact of group bias in the general LTR regime, following prior work on unbiased \ac{LTR}~\citep{joachims2017unbiased, vardasbi2020when, vardasbi2021mixture},
we  the Yahoo! Webscope~\citep{Chapelle2011} and MSLR-WEB30k~\citep{qin2013introducing} datasets that are represented by query-document feature vectors of lengths 501 and 131, respectively, and both have graded relevance labels from $0$ to $4$.
% To binarize the labels, we follow prior work and take grades $3$ and $4$ as relevant and lower grades as non-relevant.
For our experiments on the tabular regime, we use the training set of the Yahoo! and MSLR datasets, with $20k$ and $19k$ queries together with $473k$ and $2.2M$ documents, respectively.
The test set of Yahoo! and MSLR dataset contains $6.7k$ and $6k$ queries together with $163k$ and $749k$ documents, respectively.
Test queries are used for our experiments on the general LTR regime.

\myparagraph{Sensitive attribute for LTR datasets.}
\input{figures_tex/stats}
We extend prior work~\citep{diaz2020evaluating, yadav2021policy, vardasbi2022probabilistic} and use a data-driven approach for selecting features as sensitive attributes and dividing items into two groups based on some threshold on that feature.
% \ali{
We notice that synthetic sensitive attributes are not a substitute for the real sensitive attributes in real-world datasets. As mentioned above, we perform experiments on four real-world datasets with provided \emph{meaningful} sensitive attributes. Our objective for using LTR datasets with synthetic sensitive attributes is to extend and complement the results on real-world sensitive attributes with $32$ new setups with various group utility and population dynamics as will be discussed shortly.
% }
Our criterion for selecting a feature as a sensitive attribute is as follows:
For each feature we divide the items in two groups based on a threshold equal to the mean minus one standard deviation of that feature.
If more than $95\%$ of queries have at least one item from both groups, we select the feature as a candidate for sensitive attribute.
Based on this criterion, we have selected features $\{  5,  88, 100, 141, 155, 264, 393, 426\}$ and $\{ 11,  14,  15, 126, 127, 130, 131$, $132\}$ from the Yahoo! and MSLR datasets, respectively.
% \footnote{Feature numbers start from 1.}
Fig.~\ref{fig:stats} gives an overview of the ratio of affected to non-affected group members in terms of population and average utility score.
In what follows we use, e.g., $\text{Yahoo!}_{[426]}$ for the Yahoo! dataset with feature number $426$ as the sensitive feature.
For each feature, we assume two groupings based on two thresholds:
\begin{enumerate*}[label=(\roman*)]
    \item mean value; and
    \item mean minus one standard deviation.
\end{enumerate*}
This yields a total of $32$ setups.

\myparagraph{Bias simulation.}
To simulate group bias, we use \cref{eq:implicitbias}, but to make the simulation more realistic, we add a normal noise to the $\beta_{\affectedscript}$ value for each query.
We experiment with two propensities $\beta_{\affectedscript}\in\{0.6,0.8\}$ and use $\sigma_{\beta}=0.1$ for the standard deviation of the normal noise.
In Sec.~\ref{sec:experiments:inaccurate}, to add to the uncertainty of the setup, we also experiment with higher noise variances of $\sigma_{\beta}\in\{0.2, 0.3\}$.
For our correction method, we found that adding a small amount of noise to the scores for breaking the ties, without swapping the order of the grades, helps to have a smoother curve for $\beta$ in \cref{eq:bestbeta}.

\myparagraph{LTR model.}
For the general LTR model (for tail queries) we use a neural network with attention and LambdaRank Loss as in~\citep{Pobrotyn2020ContextAwareLT}.

\vspace{-1em}
\subsection{Impact of Group Bias}
\label{sec:experiments:impact}
First, we show that group bias, on both tabular and LTR regimes, consistently has a negative impact on the ranking quality and fairness metrics.
To do so, we run experiments on two datasets, namely Yahoo! and MSLR, each with $8$ different features as the sensitive attribute, and two  thresholds for separating the groups (see the ``Sensitive attribute'' paragraph in \cref{sec:setup}).
This gives us a total of $32$ different setups.
For each setup, we simulate the attractiveness probabilities with $\beta_{\affectedscript}\in\{0.8, 0.6\}$ as \emph{mild} and \emph{severe} cases of group bias, and compare the NDCG@10, \ac{DTR}, and \ac{EEL} metrics against the full information case.
Our experiments with other values of bias parameters led to consistent results.

\input{figures_tex/datasets}

\cref{fig:datasetsndcg} shows a summary of the impact of group bias on the ranking performance of both tabular and LTR regimes.
% Here we show the ratio of NDCG@10 between the biased and the full information cases.
These results confirm that group bias degrades the ranking quality of the affected group in the tabular regime and this damage is also reflected in the LTR output, trained over the biased training labels.
As a result of pushing down the relevant members of the affected group, the non-affected group gains ranking quality, i.e., the NDCG of the non-affected group with group bias is higher than the full information case.
However, the overall ranking is worsened with group bias.
Comparing the tabular (left) and LTR output (right) plots in \cref{fig:datasetsndcg}, we observe that increasing the severity of group bias from $0.8$ to $0.6$, affects the tabular regime more.
This may be because the impact of group bias on the LTR outputs is indirect.\looseness=-1

\input{figures_tex/cor_tabular}
\input{figures_tex/cor_ltr}

\cref{fig:datasetsfairness} shows a summary of the change in fairness metrics of both tabular and LTR regimes as a result of group bias.
% Here, we show the ratio of the target exposure for the DTR and EEL metrics between the biased and full information cases.
For example, a value of $\rho(\text{DTR})=0.5$ in the left plot means that on average, the target (i.e., ideal) exposure computed by the biased attractiveness scores differs from the true target exposure computed in the full information case by a factor of $0.5$.
Similarly, a value of $\Delta(\text{EEL})=3$ in the right plot means that on average, the target exposure of the biased case has an $\ell_2$ distance of $3$ to the full information target exposure.
These results confirm that group bias changes the target exposure in the tabular regime and this change is reflected in the LTR output, trained over the biased training labels.
Consequently, when the system distributes the exposure according to the target exposure to make a ranking fair, if the scores are suffering from group bias, the result is not truly fair.

\vspace{-1em}
\subsection{Amortized Correction}
% So far, our theoretical results in \cref{sec:theory} and empirical results in \cref{sec:experiments:impact} confirm the negative impact of group bias on ranking and fairness metrics.
% To correct for this bias, we have proposed an amortized correction method (\cref{sec:correction}).
In the next set of experiments, we show the effectiveness of our proposed correction method in \cref{sec:correction} for compensating for the negative effect of group bias.
\cref{tab:tabularall} compares the ranking quality, in terms of NDCG@10, as well as two merit-based fairness metrics, DTR and EEL, between the biased and corrected cases in the tabular regime.
% \ali{
Note here that in the tabular regime, it is assumed that accurate relevance estimations, up to the group bias, are available, meaning that the unbiased metrics get their ideal values, i.e., unbiased NDCG@10 equals $1$. When the corrected NDCG@10 reaches $1$ (e.g., MovieLens with both bias parameter values), it means a full recovery from group bias.
% }
In all setups, our correction method improves the ranking quality and fairness metrics over the biased case.
With some exceptions for the ranking quality with mild group bias ($\beta=0.8$), the improvements are significant.
For each bias parameter value, we have also included the estimated value obtained from~\cref{eq:bestbeta}.
We observe that our estimated bias values, i.e., $\hat{\beta}_{\affectedscript}$, are close to their corresponding true values $\beta$.

We further analyze the effectiveness of our correction method in the general LTR regime.
\cref{tab:ltrall} contains the comparison of ranking quality and fairness metrics between the biased and corrected cases in our tested LTR datasets.
We also report the full information case in the table.
Here, we only report the results for one sensitive attribute for each dataset, noting that the results for other sensitive attributes lead to similar observations.
Similar to the tabular regime, here we also observe performance improvements as a result of our correction method, compared to the biased case.
Except for the DTR metric in MSLR, all the improvements are significant with $p<0.001$.
Compared with the full information case, we observe that in the Yahoo! dataset, our correction method leads to full recovery of NDCG@10, while in the MSLR dataset, there remains a slight gap toward the full information quality.
One reason for this difference could be the distribution of relevant items in the affected and non-affected groups:
In $\text{Yahoo!}_{[426]}$ the ratio between the mean relevance of items in $G_{\affectedscript}$ to $G_{\notaffectedscript}$ is $1.05$, whereas the same quantity in $\text{MSLR}_{[127]}$ is $2.21$.
Therefore, the assumption of similar utility score distributions for both groups is closer to reality in $\text{Yahoo!}_{[426]}$ than in $\text{MSLR}_{[127]}$.
Similarly to NDCG, we observe that DTR and EEL are almost fully recovered from group bias in the Yahoo! dataset, but not in the MSLR dataset.

Finally, Fig.~\ref{fig:datasetscorrectedndcg} shows a summary of the ranking quality of biased (left) and corrected (right) utility scores in the tabular regime on all $32$ setups of the LTR datasets mentioned in Sec.~\ref{sec:experiments:impact}.
In all but two cases, we observe that our correction method effectively improves the ranking quality over the biased case and achieves NDCG@10 close to $1$.
The two outlier cases correspond to $\text{Yahoo!}_{[5]}$ ({for each feature, two different thresholds for separating the groups are used}), where the ratio between the average utility of the affected group and the non-affected group is as low as $0.3$.
This is the same outlier as in Fig.~\ref{fig:stats}.
% As our correction method is based on the same distribution assumption, this severe violation leads to inferred propensities that are noticeably lower than the actual propensity (i.e., $0.49$ and $0.36$ instead of $0.8$ and $0.6$).
It is worth mentioning that in a slightly less severe violation of the same distribution assumption, i.e., $\text{MSLR}_{[130]}$ with a utility ratio of $0.45$, our correction method is able to improve the ranking quality over the biased case.
One interesting future direction would be to find out if this phenomenon, i.e., having the true average utility of the underrepresented group considerably lower than the other group, happens in real-world settings and how to correct for the bias in such cases.

\begin{remark}
    This paper introduces a novel bias paradigm. Our study emphasizes that even the supposedly ideal fair rankings are not actually fair when group bias is unaddressed (RQ1); and that employing amortized correction yields more robust results (RQ2). Baseline comparisons are only pertinent to RQ2. We compare our amortized correction method to the per-query method which is the \emph{only} existing method for group bias measurement and correction. It is noteworthy that our method often achieves full recovery, signifying the highest conceivable effectiveness in corresponding cases.\looseness=-1
\end{remark}

\input{figures_tex/datasets_corrected}
% \vspace{-1em}
\subsection{Ablation Study}
\subsubsection{Impact of cluster size on correction}
\input{figures_tex/clustersize}
In \cref{sec:correction:amortized} we argued against measuring group propensity for each query.
Here, we analyze the impact of cluster size on the correction method.
% We start from the extreme case of one query per cluster and increase the cluster size until the ranking quality converges.
\cref{fig:clustersize} shows the ranking quality of the corrected scores as a function of the cluster size.
The overall ranking quality (red line) improves as the cluster size grows and it converges to its final value at around a cluster size of $10$.
For severe group bias ($\beta_{\affectedscript}=0.6$), which we omit due to space restrictions, the same pattern is observed, but with a convergence point of $100$.
In both cases, using a cluster size below the convergence point leads to corrected rankings that are even worse than the biased ranking.
Comparing the ranking quality of the affected group (black line) with the non-affected group (golden line), we observe that smaller clusters result in over-compensation of group bias.
The reason is revealed in \cref{fig:clustersize}(b): for smaller cluster sizes, the inferred propensity is under-estimated, leading to larger corrected scores for the affected group members.
% Consequently, the scores of the affected group are boosted more than they really should.
One other interesting observation in \cref{fig:clustersize}(b) is the high variance of the inferred propensity for small clusters (issue (i) in \cref{sec:correction:amortized}).\looseness=-1

\subsubsection{Impact of clustering accuracy}
\label{sec:experiments:inaccurate}
\input{figures_tex/inaccurate}
Finally, we address the challenge of inaccurate clustering of queries based on their group propensity that we raised at the end of \cref{sec:correction:amortized}.
The main goal of the following sets of experiments is to show that our correction method, even when accurate clustering of queries is not available, is still effective in improving the ranking quality over the biased case.
To confirm this, we add to the uncertainty of our simulation setup in two different ways:
\begin{enumerate*}[label=(\roman*)]
    \item \emph{Higher variance}: We increase the variance of the group propensity when simulating attractiveness. We consider $\sigma_{\beta}\in\{0.2, 0.3\}$. 
    % With the increased variance, the group propensity of queries can go far from the mean value, and, as we correct the queries with a single inferred value for the propensity, the probability of a mismatch between the actual propensity and the propensity used for correction increases.
    \item \emph{Two modes}: Instead of using a unimodal normal distribution to simulate group propensity, we use a mixture model with two modes $\{0.6, 0.8\}$. This means that for half of the queries, the group propensity follows a normal distribution with a mean of $0.6$ while for the other half, the mean is $0.8$ and during inference, we are not given the information about which query belongs to which mode.
\end{enumerate*}
\cref{fig:inaccurate} shows the ranking quality of the corrected scores w.r.t. different cluster sizes.
% in the increased uncertainty setups described above.
In both plots, we observe that increasing the variance of the simulated group propensity both increases the negative impact of group bias on ranking (dotted lines) and makes it harder to correct for the bias (solid lines).
The important result of these experiments, however, is that even though the uncertainty about group propensity is high, our amortized correction method almost always improves the ranking quality over the biased case.
Note that in all setups, per-query correction as well as clusters with a small size lead to worse ranking qualities than the biased scores.
Interestingly, when there are two modes of group propensity (right plot), our correction method, oblivious to the mode membership and assuming a fixed propensity, is able to correct the scores and achieve a ranking performance higher than the biased case.

%% file: figures_tex/stats.tex
\begin{figure}[t]
    \centering
    {
    \def\arraystretch{.5}
    \begin{tabular}[]{@{}l@{\hspace{0.2em}}c@{}}
        &
        \includegraphics[scale=0.4]{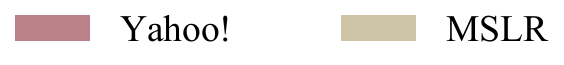}  
        \\
        \rotatebox[origin=lt]{90}{\hspace{3em} \small $G_{\affectedscript}/G_{\notaffectedscript}$ }
        &
        \includegraphics[width=0.22\textwidth]{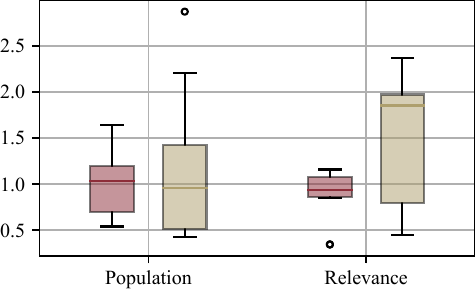}
        
    \end{tabular}
    }
    \caption{Ratio of affected to non-affected group members in terms of population and average utility score (relevance) for different sensitive attributes in Yahoo! and MSLR datasets.
    \vspace{-0.3em}
    }
    \label{fig:stats}
\end{figure}

%% file: figures_tex/datasets.tex
\begin{figure}[t]
    \centering
    {
    \def\arraystretch{.5}
    \begin{tabular}[]{@{}l@{\hspace{0.2em}}c@{\hspace{0.1em}}c@{}}
        &
        \multicolumn{2}{c}{
        \hspace{0.1em}
        \includegraphics[scale=0.38]{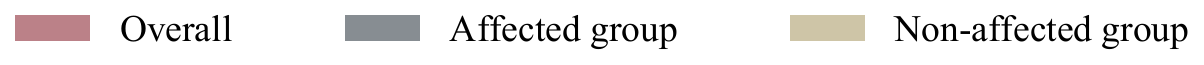}  
        } \\
        & Tabular (training labels)
        & LTR output \\
        \rotatebox[origin=lt]{90}{\hspace{0.7em} \small$\Delta$ NDCG@10}
        &
        \includegraphics[width=0.2\textwidth]{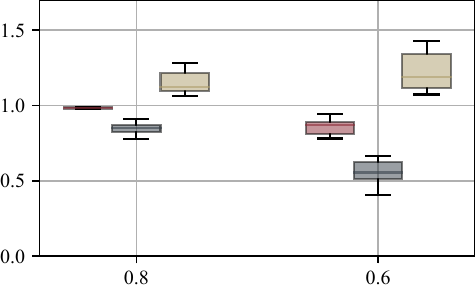}
        &
        \includegraphics[width=0.2\textwidth]{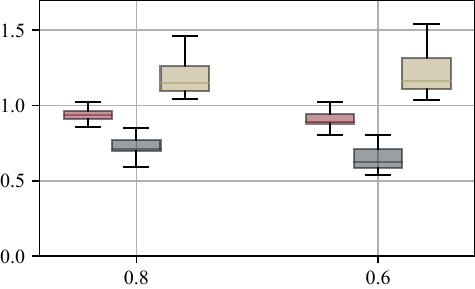} \\
        & \small Group propensity
        & \small Group propensity
        
    \end{tabular}
    }
    \caption{The impact of group bias on ranking quality for the Yahoo! and MSLR datasets with different sensitive attributes.
    \vspace{-2em}
    }
    \label{fig:datasetsndcg}
\end{figure}

\begin{figure}[t]
    \centering
    {
    \def\arraystretch{.5}
    \begin{tabular}[]{@{}l@{\hspace{0.2em}}c@{\hspace{1em}}l@{\hspace{0.2em}}c@{}}
    & \multicolumn{3}{c}{\includegraphics[scale=0.38]{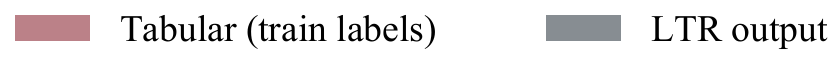}} \\
        & DTR $\uparrow$ (ideal: $1$)
        & & EEL $\downarrow$ (ideal: $0$) \\
        \rotatebox[origin=lt]{90}{\hspace{2.2em} \small $\rho(\text{DTR})$}
        &
        \includegraphics[width=0.21\textwidth]{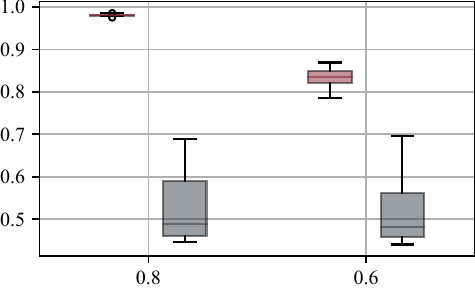}
        & \rotatebox[origin=lt]{90}{\hspace{2.2em} \small $\Delta(\text{EEL})$}
        &
        \includegraphics[width=0.21\textwidth]{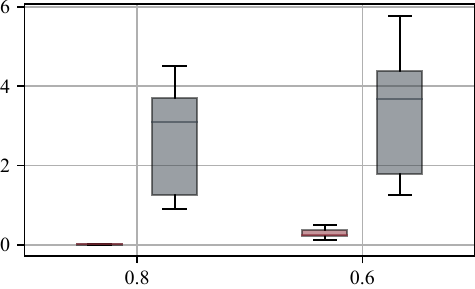} 
        \\
        & \small Group propensity
        & & \small Group propensity
        
    \end{tabular}
    }
    \caption{The impact of group bias on fairness metrics for the Yahoo! and MSLR datasets with different sensitive attributes.
    \vspace{-2em}
    }
    \label{fig:datasetsfairness}
\end{figure}

%% file: figures_tex/cor_tabular.tex
% \begin{figure}[t]
%     \centering
%     \begin{tabular}[]{@{}l@{\hspace{0.2em}}c@{\hspace{0.1em}}c@{}}
%         &
%         \multicolumn{2}{c}{
%         \hspace{0.1em}
%         \includegraphics[scale=0.38]{figures/correction_legend}  
%         } \\
%         & $\text{Yahoo!}_{[426]}$
%         & $\text{MSLR}_{[127]}$ \\
%         \rotatebox[origin=lt]{90}{\hspace{0.6em} \small $\rho(\text{DTR})\uparrow$}
%         &
%         \includegraphics[width=0.22\textwidth]{figures/yahoo_ltr_DTR}
%         &
%         \includegraphics[width=0.22\textwidth]{figures/mslr_ltr_DTR} \\
%         \rotatebox[origin=lt]{90}{\hspace{0.6em} \small $\Delta(\text{EEL})\downarrow$}
%         &
%         \includegraphics[width=0.22\textwidth]{figures/yahoo_ltr_EEL}
%         &
%         \includegraphics[width=0.22\textwidth]{figures/mslr_ltr_EEL} \\
%         % \rotatebox[origin=lt]{90}{\hspace{0.6em} \small $\Delta(\text{EEL})\downarrow$}
%         % &
%         % \includegraphics[width=0.22\textwidth]{figures/mslr_tab_EEL}
%         % &
%         % \includegraphics[width=0.22\textwidth]{figures/mslr_ltr_EEL} \\
%         & \small Group propensity
%         & \small Group propensity
        
%     \end{tabular}
%     \caption{The impact of our amortized group bias correction on DTR (top) and EEL (bottom) fairness metric in the general LTR regime. 
%     % Top: $\text{Yahoo!}_{[426]}$; bottom: $\text{MSLR}_{[127]}$.
%     }
%     \label{fig:eel}
% \end{figure}

\begin{table*}[t]
    \centering
    \caption{The impact of our amortized group bias correction on ranking and fairness metrics in the tabular regime. $\hat{\beta}_{\affectedscript}$ shows the estimated value of the bias parameter as in Eq.~\eqref{eq:bestbeta}. For each metric, the columns ``\textit{B}'' and ``C'' show the ``\textit{Biased}'' and ``Corrected'' performances, respectively. Superscripts $^*$ indicate a significant improvement over the biased case with $p<0.001$.}

    \begin{tabular}{l@{\hspace{1em}}D BA BA BA 
    D@{\hspace{1em}}BA BA BA}
    \toprule
        & \multicolumn{7}{c}{$\beta=0.8$}
        & \multicolumn{7}{c}{$\beta=0.6$}
        \\
        \cmidrule(lr){2-8}\cmidrule(lr){9-15}
        & \cellcolor{white}\multirow{2}{*}{$\hat{\beta}_{\affectedscript}$}
        & \multicolumn{2}{c}{\small $\text{NDCG}@10\uparrow$}
        & \multicolumn{2}{c}{\small $\rho(\text{DTR})\uparrow$}
        & \multicolumn{2}{c}{\small $\Delta(\text{EEL})\downarrow$} 
        
        & \cellcolor{white}\multirow{2}{*}{$\hat{\beta}_{\affectedscript}$}
        & \multicolumn{2}{c}{\small $\text{NDCG}@10\uparrow$} 
        & \multicolumn{2}{c}{\small $\rho(\text{DTR})\uparrow$}
        & \multicolumn{2}{c}{\small $\Delta(\text{EEL})\downarrow$} 
        \\
        \cmidrule(lr){3-4}\cmidrule(lr){5-6}\cmidrule(lr){7-8}
        \cmidrule(lr){10-11}\cmidrule(lr){12-13}\cmidrule(lr){14-15}
        & \cellcolor{white}
        & B & C & B & C & B & C
        & \cellcolor{white}
        & B & C & B & C & B & C
        \\
\midrule
$\text{Yahoo!}_{[426]}$    &               0.825 &  0.987 &  0.996$^*$ &    0.820 &  0.955$^*$ &    0.447 &  0.120$^*$ &               0.626 &  0.885 &  0.988$^*$ &    0.641 &  0.941$^*$ &    1.809 &  0.172$^*$ \\
$\text{MSLR}_{[127]}$      &               0.843 &  0.975 &  0.991$^*$ &    0.813 &  0.948$^*$ &    1.687 &  0.308$^*$ &               0.648 &  0.780 &  0.966$^*$ &    0.627 &  0.926$^*$ &    7.146 &  0.471$^*$ \\
IIT-JEE                    &               0.727 &  0.989 &      0.991\phantom{$^*$} &    0.799 &  0.906$^*$ &    0.504 &  0.341$^*$ &               0.547 &  0.970 &  0.985$^*$ &    0.600 &  0.901$^*$ &    1.401 &  0.410$^*$ \\
$\text{MovieLens}_{[Co.]}$ &               0.822 &  1.000 &      1.000\phantom{$^*$} &    0.800 &  0.962$^*$ &    1.101 &  0.513$^*$ &               0.612 &  0.998 &  1.000$^*$ &    0.602 &  0.958$^*$ &    7.113 &  1.908$^*$ \\
$\text{MovieLens}_{[BO]}$  &               0.781 &  1.000 &      1.000\phantom{$^*$} &    0.799 &  0.974$^*$ &    2.330 &  0.895$^*$ &               0.579 &  0.994 &  1.000$^*$ &    0.600 &  0.961$^*$ &   24.800 &  2.831$^*$ \\
TREC 2019                  &               0.838 &  0.997 &      1.000\phantom{$^*$} &    0.888 &  0.954$^*$ &    0.041 &  0.020$^*$ &               0.634 &  0.986 &  0.999$^*$ &    0.771 &  0.937$^*$ &    0.129 &  0.028$^*$ \\
TREC 2020                  &               0.821 &  0.995 &      0.999\phantom{$^*$} &    0.815 &  0.954$^*$ &    0.356 &  0.114$^*$ &               0.614 &  0.945 &  0.995$^*$ &    0.627 &  0.941$^*$ &    1.152 &  0.137$^*$ \\
\bottomrule
 
    \end{tabular}
    \label{tab:tabularall}
\end{table*}

%% file: figures_tex/cor_ltr.tex
% \begin{figure}[t]
%     \centering
%     \begin{tabular}[]{@{}l@{\hspace{0.2em}}c@{\hspace{0.1em}}c@{}}
%         &
%         \multicolumn{2}{c}{
%         \hspace{0.1em}
%         \includegraphics[scale=0.38]{figures/correction_legend}  
%         } \\
%         % & Tabular (training labels)
%         % & LTR output \\
%         & $\text{Yahoo!}_{[426]}$
%         & $\text{MSLR}_{[127]}$ \\
%         \rotatebox[origin=lt]{90}{\hspace{0.2em} \small NDCG@10}
%         &
%         \includegraphics[width=0.22\textwidth]{figures/yahoo_ltr_ndcg}
%         &
%         \includegraphics[width=0.22\textwidth]{figures/mslr_ltr_ndcg} \\
%         % \rotatebox[origin=lt]{90}{\hspace{0.2em} \small NDCG@10}
%         % &
%         % \includegraphics[width=0.22\textwidth]{figures/mslr_tab_ndcg}
%         % &
%         % \includegraphics[width=0.22\textwidth]{figures/mslr_ltr_ndcg} \\
%         & \small Group propensity
%         & \small Group propensity
        
%     \end{tabular}
%     \caption{The impact of our amortized group bias correction on ranking performance in the general LTR regime. 
%     % Top: $\text{Yahoo!}_{[426]}$; bottom: $\text{MSLR}_{[127]}$.
%     }
%     \label{fig:ndcg}
% \end{figure}

\begin{table*}[t]
    \centering
    \caption{The impact of our amortized group bias correction on ranking and fairness metrics in general LTR regime. For each metric, the columns ``\textit{B}'' and ``C'' show the ``\textit{Biased}'' and ``Corrected'' performances, respectively. 
    % Superscripts $^*$ indicate a significant improvement over the biased case with $p<0.001$.
    }

    \begin{tabular}{l@{\hspace{1em}}ccc@{\hspace{1em}}BA BA BA 
    @{\hspace{1em}}BA BA BA}
    \toprule
        & \multicolumn{3}{c}{Full info.}
        & \multicolumn{6}{c}{$\beta=0.8$}
        & \multicolumn{6}{c}{$\beta=0.6$}
        \\
        \cmidrule(lr){2-4}\cmidrule(lr){5-10}\cmidrule(lr){11-16}
        & \cellcolor{white}\multirow{2}{*}{\rotatebox[origin=lt]{55}{\small $\text{NDCG}$}}
        & \cellcolor{white}\multirow{2}{*}{\rotatebox[origin=lt]{55}{\small$\rho(\text{DTR})$}}
        & \cellcolor{white}\multirow{2}{*}{\rotatebox[origin=lt]{55}{\small$\Delta(\text{EEL})$}}

        & \multicolumn{2}{c}{\small $\text{NDCG}@10\uparrow$}
        & \multicolumn{2}{c}{\small $\rho(\text{DTR})\uparrow$}
        & \multicolumn{2}{c}{\small $\Delta(\text{EEL})\downarrow$} 
        
        % & \cellcolor{white}\multirow{2}{*}{$\hat{\beta}_{\affectedscript}$}
        & \multicolumn{2}{c}{\small $\text{NDCG}@10\uparrow$} 
        & \multicolumn{2}{c}{\small $\rho(\text{DTR})\uparrow$}
        & \multicolumn{2}{c}{\small $\Delta(\text{EEL})\downarrow$} 
        \\
        \cmidrule(lr){5-6}\cmidrule(lr){7-8}\cmidrule(lr){9-10}
        \cmidrule(lr){11-12}\cmidrule(lr){13-14}\cmidrule(lr){15-16}
        % & \cellcolor{white} & \cellcolor{white} & \cellcolor{white}
        & & &
        & B & C & B & C & B & C
        & B & C & B & C & B & C
        \\
\midrule
$\text{Yahoo!}_{[426]}$ &       0.649 &    0.477 &    0.673 &       0.615 &  0.645$^*$ &    0.319 &  0.428$^*$ &    1.970 &  0.741$^*$ &       0.580 &  0.645$^*$ &    0.313 &  0.456$^*$ &    2.711 &  0.728$^*$ \\

$\text{MSLR}_{[127]}$   &       0.320 &    0.681 &    1.779 &       0.283 &  0.313$^*$ &    0.675 &      0.684\phantom{$^*$} &    4.380 &  2.016$^*$ &       0.265 &  0.309$^*$ &    0.671 &      0.683\phantom{$^*$} &    5.803 &  2.147$^*$ \\
\bottomrule
 
    \end{tabular}
    \label{tab:ltrall}
\end{table*}

%% file: figures_tex/datasets_corrected.tex
\begin{figure}[t]
    \centering
    {
    \def\arraystretch{.5}
    \begin{tabular}[]{@{}l@{\hspace{0.2em}}c@{\hspace{0.1em}}c@{}}
        &
        \multicolumn{2}{c}{
        \hspace{0.1em}
        \includegraphics[scale=0.38]{figures/datasets_ndcg_legend}  
        } \\
        & Biased
        & Corrected \\
        \rotatebox[origin=lt]{90}{\hspace{1.0em} \small NDCG@10}
        &
        \includegraphics[width=0.2\textwidth]{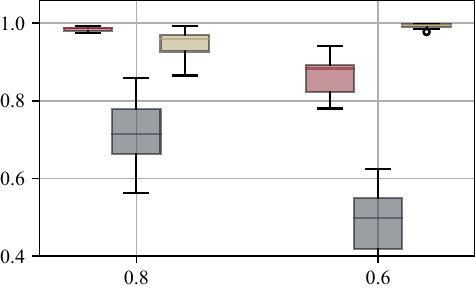}
        &
        \includegraphics[width=0.2\textwidth]{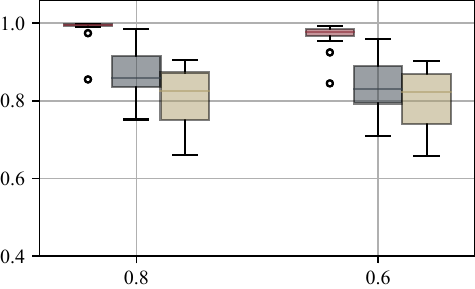} \\
        & \small Group propensity
        & \small Group propensity
        
    \end{tabular}
    }
    \caption{The ranking performance of biased (left) and corrected (right) scores for the Yahoo! and MSLR datasets with different sensitive attributes.\vspace{-0.4em}}
    \label{fig:datasetscorrectedndcg}
\end{figure}

%% file: figures_tex/clustersize.tex
% \begin{figure*}[t]
%     \centering
%     \begin{tabular}[]{@{}l@{\hspace{0.2em}}c@{\hspace{0.2em}}c@{\hspace{0.2em}}c@{\hspace{0.1em}}c@{}}
%         & \multicolumn{3}{c}{\includegraphics[scale=0.45]{figures/clustersize_AN_legend}} \\
%         & $\sigma_{\beta}=0.1$
%         & $\sigma_{\beta}=0.2$
%         & $\sigma_{\beta}=0.3$ \\
%         \rotatebox[origin=lt]{90}{\hspace{2.6em} \small NDCG@10}
%         &
%         \includegraphics[width=0.3\textwidth]{figures/clustersize_AN_0.8_0.1}
%         &
%         \includegraphics[width=0.3\textwidth]{figures/clustersize_AN_0.8_0.2}
%         &
%         \includegraphics[width=0.3\textwidth]{figures/clustersize_AN_0.8_0.3}
%          \\
%         \rotatebox[origin=lt]{90}{\hspace{2.6em} \small NDCG@10}
%         &
%         \includegraphics[width=0.3\textwidth]{figures/clustersize_AN_0.6_0.1}
%         &
%         \includegraphics[width=0.3\textwidth]{figures/clustersize_AN_0.6_0.2}
%         &
%         \includegraphics[width=0.3\textwidth]{figures/clustersize_AN_0.6_0.3}
%          \\
%         & \small Cluster size
%         & \small Cluster size
%         & \small Cluster size
%     \end{tabular}
%     \caption{The impact of cluster size and variance of group propensity on the amortized correction for group bias on ranking performance for the $\text{Yahoo!}_{[426]}$ dataset. Top: $\beta=0.8$; bottom: $\beta=0.6$.}
%     \label{fig:clustersize}
% \end{figure*}
{
\begin{figure}[t]
    \centering
    {
    \def\arraystretch{.5}
    \begin{tabular}[]{@{}l@{\hspace{0.2em}}c@{\hspace{0.7em}}c@{}}
        & \multicolumn{2}{c}{\includegraphics[scale=0.45]{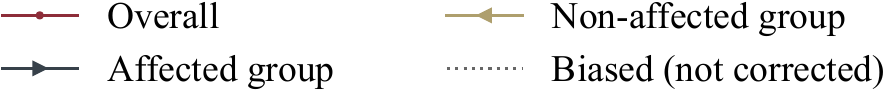}} 
        % & \includegraphics[scale=0.245]{figures/beta_correct_legend} 
        \\
        % & $\beta_{\affectedscript}=0.8$
        % & $\beta_{\affectedscript}=0.6$ \\
        \rotatebox[origin=lt]{90}{\hspace{1.6em} \small NDCG@10}
        &
        \includegraphics[width=0.2\textwidth]{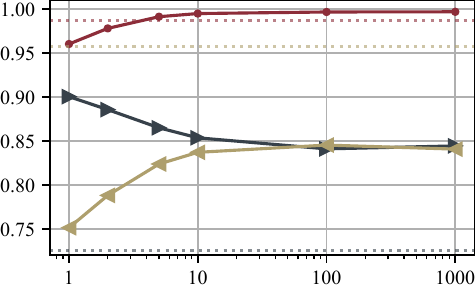}
        &
        \includegraphics[width=0.2\textwidth]{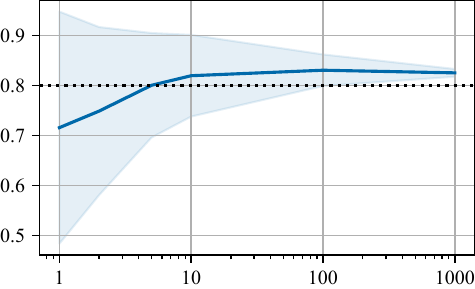}
         \\
        & \small Cluster size
        & \small Cluster size
        \\
        & a) Ranking quality
        & b) Inferred $\beta$ values
    \end{tabular}
    }
    \caption{The impact of cluster size of group propensity on the amortized correction for group bias ($\beta_{\affectedscript}=0.8$) for the $\text{Yahoo!}_{[426]}$ dataset. Ranking quality of corrected scores (a), and accuracy of the inferred group propensity (b).
    \vspace{-.1em}
    }
    \label{fig:clustersize}
\end{figure}
}

%% file: figures_tex/inaccurate.tex
{
\begin{figure}[t]
    \centering
    {
    \def\arraystretch{.5}
    \begin{tabular}[]{@{}l@{\hspace{0.2em}}c@{\hspace{0.3em}}c@{}}
        & \multicolumn{2}{c}{
        \includegraphics[scale=0.45]{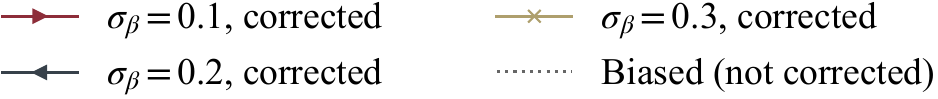}  }
        \\
        & $\mu_{\beta}=0.8$
        & $\mu_{\beta}\in\{0.6,0.8\}$ 
        \\
        \rotatebox[origin=lt]{90}{\hspace{2em} \small NDCG@10}
        &
        \includegraphics[width=0.2\textwidth]{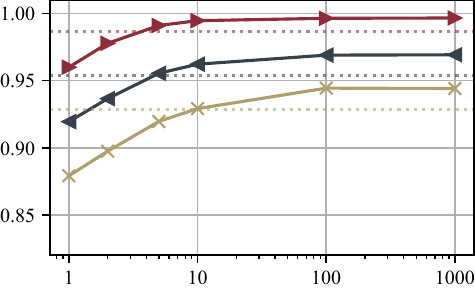}
        &
        \includegraphics[width=0.2\textwidth]{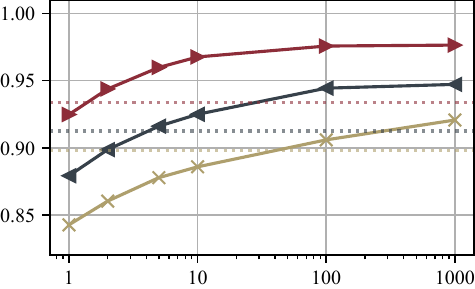}
        \\
        & \small Cluster size
        & \small Cluster size
        
    \end{tabular}
    }
    \caption{Effectiveness of amortized correction method when accurate clustering is not available. $\text{Yahoo!}_{[426]}$ dataset.
\vspace{-0.1em}
    }
    \label{fig:inaccurate}
\end{figure}
}

%% file: sections/07-conclusion.tex
% !TEX root = ../main.tex
% \vspace{-0.2em}
\section{Conclusion and Future Work}
\label{sec:conclusion} 
We have addressed group membership bias, which is based on the observation that a user's perception of an item's group membership may affect their judgment about the utility of an item.
We have provided extensive theoretical and empirical analyses of the impact of group bias on the ranking quality and two fairness of exposure metrics, DTR and EEL.
% Our theoretical results show the relation between the change in different metrics and the severity of group bias.
By utilizing an auxiliary variable $\nu$ as the fraction of affected relevant items that are still as attractive as the non-affected relevant items,
% Regarding the ranking quality of tabular models, we showed that the change of \ac{NDCG} as a result of group bias can be approximated by a linear function of the mean value of $\nu$.
% Similarly, the ratio between the DTR target exposure of biased and full information cases, equals $\nu$.
% We also provide an approximation, logarithmic in $\nu$, of the change in the EEL target exposure as a result of group bias.
we have shown that, in the presence of group bias, NDCG and DTR change linearly with $\nu$, while the change in EEL has a more complex form in terms of $\nu$.

Correcting for group bias, which is a type of \emph{content-based} bias, is not as easy as \emph{context-based} types of bias such as position and trust bias.
To measure group bias, assumptions based on fairness constraints should be made about the utility distribution of different groups.
However, such assumptions can potentially make the equity-based notion of fairness meaningless.
% For example, assuming that both groups have the same utility score distribution leads to the corrected utilities for the affected group that, when averaged, are always the same as the average non-affected utility.
% This contradicts the premise of equity that different groups may have different utility values and the exposure should be distributed according to the utility.
Amortized correction for the group bias is our solution to this issue, as global equality does not contradict local equity.
We have experimentally confirmed that our correction method, when its assumptions are met, is able to fully recover the scores suffering from group bias.
% , in the sense that the ranking and fairness metrics after correction achieve the values of the full information case.

% Our amortized correction, however, raises a challenge of its own, as it is not easy to cluster queries with almost the same group propensity without the knowledge of group propensity.
% Our experiments have shown that even when an accurate clustering of queries is not available, loosely clustering the queries for the amortized correction still leads to better rankings compared to the biased scores.
% More interestingly, per-query correction as well as clusters of small size lead to worse ranking qualities than the biased case.\looseness=-1

There are several future directions to this study.
% Here, we analyzed a multiplicative model of group bias.
% , and our theoretical and empirical results are based on this formulation.
One way to extend our results is to 
% consider more complex models for group bias.
% Another future direction is to 
propose measurement and correction methods that perform better with increased uncertainty of group propensity.
Moreover, as group bias is based on users' perception of group membership, it can change over time.
Analyzing group bias in a dynamic setting is therefore another direction.
This study deals with the impact of group bias on fair exposure and, hence, we only consider the so-called treatment-based fairness metrics.
In contrast, some studies focus on impact-based fairness~\citep{saito2022fair}, where the objective is to make sure that items receive a fair amount of impact, e.g., clicks.
While our work suggests a way to correct for group bias in the historical clicks in order to make the exposure in future rankings fair, a next direction would be to account for group bias when optimizing for impact-based fairness.

%% file: main.bbl

\begin{thebibliography}{58}

%%% the Web address.

\ifx \showCODEN    \undefined \def \showCODEN     #1{\unskip}     \fi
\ifx \showDOI      \undefined \def \showDOI       #1{#1}\fi
\ifx \showISBNx    \undefined \def \showISBNx     #1{\unskip}     \fi
\ifx \showISBNxiii \undefined \def \showISBNxiii  #1{\unskip}     \fi
\ifx \showISSN     \undefined \def \showISSN      #1{\unskip}     \fi
\ifx \showLCCN     \undefined \def \showLCCN      #1{\unskip}     \fi
\ifx \shownote     \undefined \def \shownote      #1{#1}          \fi
\ifx \showarticletitle \undefined \def \showarticletitle #1{#1}   \fi
\ifx \showURL      \undefined \def \showURL       {\relax}        \fi
% invisible to TeX
\providecommand\bibfield[2]{#2}
\providecommand\bibinfo[2]{#2}
\providecommand\natexlab[1]{#1}
\providecommand\showeprint[2][]{arXiv:#2}

\bibitem[\protect\citeauthoryear{Ai, Bi, Luo, Guo, and Croft}{Ai
  et~al\mbox{.}}{2018}]%
        {ai2018unbiased}
\bibfield{author}{\bibinfo{person}{Qingyao Ai}, \bibinfo{person}{Keping Bi},
  \bibinfo{person}{Cheng Luo}, \bibinfo{person}{Jiafeng Guo}, {and}
  \bibinfo{person}{W~Bruce Croft}.} \bibinfo{year}{2018}\natexlab{}.
\newblock \showarticletitle{Unbiased Learning to Rank with Unbiased Propensity
  Estimation}. In \bibinfo{booktitle}{\emph{The 41st International ACM SIGIR
  Conference on Research \& Development in Information Retrieval}}. ACM,
  \bibinfo{pages}{385--394}.
\newblock


\bibitem[\protect\citeauthoryear{Balagopalan, Jacobs, and Biega}{Balagopalan
  et~al\mbox{.}}{2023}]%
        {Balagopalan2023therole}
\bibfield{author}{\bibinfo{person}{Aparna Balagopalan},
  \bibinfo{person}{Abigail~Z. Jacobs}, {and} \bibinfo{person}{Asia~J. Biega}.}
  \bibinfo{year}{2023}\natexlab{}.
\newblock \showarticletitle{The Role of Relevance in Fair Ranking}. In
  \bibinfo{booktitle}{\emph{Proceedings of the 46th International ACM SIGIR
  Conference on Research and Development in Information Retrieval}}
  \emph{(\bibinfo{series}{SIGIR '23})}. \bibinfo{publisher}{Association for
  Computing Machinery}, \bibinfo{address}{New York, NY, USA},
  \bibinfo{pages}{2650–2660}.
\newblock
\showISBNx{9781450394086}
\urldef\tempurl%
\url{https://doi.org/10.1145/3539618.3591933}
\showDOI{\tempurl}


\bibitem[\protect\citeauthoryear{Baswana, Chakrabarti, Chandran, Kanoria, and
  Patange}{Baswana et~al\mbox{.}}{2019}]%
        {baswana2019centralized}
\bibfield{author}{\bibinfo{person}{Surender Baswana},
  \bibinfo{person}{Partha~Pratim Chakrabarti}, \bibinfo{person}{Sharat
  Chandran}, \bibinfo{person}{Yashodhan Kanoria}, {and}
  \bibinfo{person}{Utkarsh Patange}.} \bibinfo{year}{2019}\natexlab{}.
\newblock \showarticletitle{Centralized Admissions for Engineering Colleges in
  India}. In \bibinfo{booktitle}{\emph{Proceedings of the 2019 ACM Conference
  on Economics and Computation}}. \bibinfo{pages}{323--324}.
\newblock


\bibitem[\protect\citeauthoryear{Biega, Diaz, Ekstrand, Feldman, and
  Kohlmeier}{Biega et~al\mbox{.}}{2021}]%
        {biega-2020-overview}
\bibfield{author}{\bibinfo{person}{Asia~J. Biega}, \bibinfo{person}{Fernando
  Diaz}, \bibinfo{person}{Michael~D. Ekstrand}, \bibinfo{person}{Sergey
  Feldman}, {and} \bibinfo{person}{Sebastian Kohlmeier}.}
  \bibinfo{year}{2021}\natexlab{}.
\newblock \showarticletitle{Overview of the {TREC} 2020 Fair Ranking Track}. In
  \bibinfo{booktitle}{\emph{The Twenty-Ninth Text REtrieval Conference
  Proceedings (TREC 2020)}}.
\newblock


\bibitem[\protect\citeauthoryear{Biega, Diaz, Ekstrand, and Kohlmeier}{Biega
  et~al\mbox{.}}{2019}]%
        {trec-fair-ranking-2019}
\bibfield{author}{\bibinfo{person}{Asia~J. Biega}, \bibinfo{person}{Fernando
  Diaz}, \bibinfo{person}{Michael~D. Ekstrand}, {and}
  \bibinfo{person}{Sebastian Kohlmeier}.} \bibinfo{year}{2019}\natexlab{}.
\newblock \showarticletitle{Overview of the TREC 2019 Fair Ranking Track}. In
  \bibinfo{booktitle}{\emph{The Twenty-Eighth Text REtrieval Conference (TREC
  2019) Proceedings}}.
\newblock


\bibitem[\protect\citeauthoryear{Biega, Gummadi, and Weikum}{Biega
  et~al\mbox{.}}{2018}]%
        {biega2018equity}
\bibfield{author}{\bibinfo{person}{Asia~J. Biega}, \bibinfo{person}{Krishna~P.
  Gummadi}, {and} \bibinfo{person}{Gerhard Weikum}.}
  \bibinfo{year}{2018}\natexlab{}.
\newblock \showarticletitle{Equity of Attention: Amortizing Individual Fairness
  in Rankings}. In \bibinfo{booktitle}{\emph{The 41st International ACM SIGIR
  Conference on Research \& Development in Information Retrieval}}
  \emph{(\bibinfo{series}{SIGIR '18})}. \bibinfo{publisher}{Association for
  Computing Machinery}, \bibinfo{address}{New York, NY, USA},
  \bibinfo{pages}{405–414}.
\newblock
\showISBNx{9781450356572}
\urldef\tempurl%
\url{https://doi.org/10.1145/3209978.3210063}
\showDOI{\tempurl}


\bibitem[\protect\citeauthoryear{Brownstein}{Brownstein}{2017}]%
        {brownstein2015implicit}
\bibfield{author}{\bibinfo{person}{Michael Brownstein}.}
  \bibinfo{year}{2017}\natexlab{}.
\newblock \showarticletitle{Implicit Bias}.
\newblock In \bibinfo{booktitle}{\emph{Stanford Encyclopedia of Philosophy}},
  \bibfield{editor}{\bibinfo{person}{Edward Zalta}} (Ed.).
\newblock


\bibitem[\protect\citeauthoryear{Celis, Mehrotra, and Vishnoi}{Celis
  et~al\mbox{.}}{2020}]%
        {Celis2020Interventions}
\bibfield{author}{\bibinfo{person}{L.~Elisa Celis}, \bibinfo{person}{Anay
  Mehrotra}, {and} \bibinfo{person}{Nisheeth~K. Vishnoi}.}
  \bibinfo{year}{2020}\natexlab{}.
\newblock \showarticletitle{Interventions for Ranking in the Presence of
  Implicit Bias}. In \bibinfo{booktitle}{\emph{Proceedings of the 2020
  Conference on Fairness, Accountability, and Transparency}}
  \emph{(\bibinfo{series}{FAT* '20})}. \bibinfo{publisher}{Association for
  Computing Machinery}, \bibinfo{address}{New York, NY, USA},
  \bibinfo{pages}{369–380}.
\newblock
\showISBNx{9781450369367}
\urldef\tempurl%
\url{https://doi.org/10.1145/3351095.3372858}
\showDOI{\tempurl}


\bibitem[\protect\citeauthoryear{Chapelle and Chang}{Chapelle and
  Chang}{2011}]%
        {Chapelle2011}
\bibfield{author}{\bibinfo{person}{Olivier Chapelle} {and} \bibinfo{person}{Yi
  Chang}.} \bibinfo{year}{2011}\natexlab{}.
\newblock \showarticletitle{Yahoo! Learning to Rank Challenge Overview}. In
  \bibinfo{booktitle}{\emph{Proceedings of the Learning to Rank Challenge}}
  \emph{(\bibinfo{series}{Proceedings of Machine Learning Research})},
  Vol.~\bibinfo{volume}{14}. \bibinfo{publisher}{PMLR}, \bibinfo{pages}{1--24}.
\newblock
\urldef\tempurl%
\url{https://proceedings.mlr.press/v14/chapelle11a.html}
\showURL{%
\tempurl}


\bibitem[\protect\citeauthoryear{Chen, Dong, Wang, Feng, Wang, and He}{Chen
  et~al\mbox{.}}{2023}]%
        {chen2023bias}
\bibfield{author}{\bibinfo{person}{Jiawei Chen}, \bibinfo{person}{Hande Dong},
  \bibinfo{person}{Xiang Wang}, \bibinfo{person}{Fuli Feng},
  \bibinfo{person}{Meng Wang}, {and} \bibinfo{person}{Xiangnan He}.}
  \bibinfo{year}{2023}\natexlab{}.
\newblock \showarticletitle{Bias and Debias in Recommender System: A Survey and
  Future Directions}.
\newblock \bibinfo{journal}{\emph{ACM Trans. Inf. Syst.}} \bibinfo{volume}{41},
  \bibinfo{number}{3}, Article \bibinfo{articleno}{67} (\bibinfo{date}{feb}
  \bibinfo{year}{2023}), \bibinfo{numpages}{39}~pages.
\newblock
\showISSN{1046-8188}
\urldef\tempurl%
\url{https://doi.org/10.1145/3564284}
\showDOI{\tempurl}


\bibitem[\protect\citeauthoryear{Chuklin, Markov, and Rijke}{Chuklin
  et~al\mbox{.}}{2015}]%
        {chuklin2015click}
\bibfield{author}{\bibinfo{person}{Aleksandr Chuklin}, \bibinfo{person}{Ilya
  Markov}, {and} \bibinfo{person}{Maarten~de Rijke}.}
  \bibinfo{year}{2015}\natexlab{}.
\newblock \bibinfo{booktitle}{\emph{Click Models for Web Search}}.
\newblock \bibinfo{publisher}{Morgan \& Claypool Publishers}.
\newblock


\bibitem[\protect\citeauthoryear{Deldjoo, Jannach, Bellogin, Difonzo, and
  Zanzonelli}{Deldjoo et~al\mbox{.}}{2023}]%
        {deldjoo2023fairness}
\bibfield{author}{\bibinfo{person}{Yashar Deldjoo}, \bibinfo{person}{Dietmar
  Jannach}, \bibinfo{person}{Alejandro Bellogin}, \bibinfo{person}{Alessandro
  Difonzo}, {and} \bibinfo{person}{Dario Zanzonelli}.}
  \bibinfo{year}{2023}\natexlab{}.
\newblock \showarticletitle{Fairness in Recommender Systems: Research Landscape
  and Future Directions}.
\newblock \bibinfo{journal}{\emph{User Modeling and User-Adapted Interaction}}
  (\bibinfo{year}{2023}), \bibinfo{pages}{1--50}.
\newblock


\bibitem[\protect\citeauthoryear{Diaz, Mitra, Ekstrand, Biega, and
  Carterette}{Diaz et~al\mbox{.}}{2020}]%
        {diaz2020evaluating}
\bibfield{author}{\bibinfo{person}{Fernando Diaz}, \bibinfo{person}{Bhaskar
  Mitra}, \bibinfo{person}{Michael~D. Ekstrand}, \bibinfo{person}{Asia~J.
  Biega}, {and} \bibinfo{person}{Ben Carterette}.}
  \bibinfo{year}{2020}\natexlab{}.
\newblock \showarticletitle{Evaluating Stochastic Rankings with Expected
  Exposure}. In \bibinfo{booktitle}{\emph{Proceedings of the 29th ACM
  International Conference on Information \& Knowledge Management}}
  \emph{(\bibinfo{series}{CIKM '20})}. \bibinfo{publisher}{Association for
  Computing Machinery}, \bibinfo{address}{New York, NY, USA},
  \bibinfo{pages}{275–284}.
\newblock
\showISBNx{9781450368599}
\urldef\tempurl%
\url{https://doi.org/10.1145/3340531.3411962}
\showDOI{\tempurl}


\bibitem[\protect\citeauthoryear{Ekstrand, McDonald, Raj, and Johnson}{Ekstrand
  et~al\mbox{.}}{2022}]%
        {trec-fair-ranking-2021}
\bibfield{author}{\bibinfo{person}{Michael~D. Ekstrand},
  \bibinfo{person}{Graham McDonald}, \bibinfo{person}{Amifa Raj}, {and}
  \bibinfo{person}{Isaac Johnson}.} \bibinfo{year}{2022}\natexlab{}.
\newblock \showarticletitle{Overview of the TREC 2021 Fair Ranking Track}. In
  \bibinfo{booktitle}{\emph{The Thirtieth Text REtrieval Conference (TREC 2021)
  Proceedings}}.
\newblock


\bibitem[\protect\citeauthoryear{Emelianov, Gast, Gummadi, and
  Loiseau}{Emelianov et~al\mbox{.}}{2020}]%
        {Emelianov2020On}
\bibfield{author}{\bibinfo{person}{Vitalii Emelianov}, \bibinfo{person}{Nicolas
  Gast}, \bibinfo{person}{Krishna~P. Gummadi}, {and} \bibinfo{person}{Patrick
  Loiseau}.} \bibinfo{year}{2020}\natexlab{}.
\newblock \showarticletitle{On Fair Selection in the Presence of Implicit
  Variance}. In \bibinfo{booktitle}{\emph{Proceedings of the 21st ACM
  Conference on Economics and Computation}} \emph{(\bibinfo{series}{EC '20})}.
  \bibinfo{publisher}{Association for Computing Machinery},
  \bibinfo{address}{New York, NY, USA}, \bibinfo{pages}{649–675}.
\newblock
\showISBNx{9781450379755}
\urldef\tempurl%
\url{https://doi.org/10.1145/3391403.3399482}
\showDOI{\tempurl}


\bibitem[\protect\citeauthoryear{Fabris, Purpura, Silvello, and Susto}{Fabris
  et~al\mbox{.}}{2020}]%
        {fabris2020gender}
\bibfield{author}{\bibinfo{person}{Alessandro Fabris}, \bibinfo{person}{Alberto
  Purpura}, \bibinfo{person}{Gianmaria Silvello}, {and}
  \bibinfo{person}{Gian~Antonio Susto}.} \bibinfo{year}{2020}\natexlab{}.
\newblock \showarticletitle{Gender Stereotype Reinforcement: Measuring the
  Gender Bias Conveyed by Ranking Algorithms}.
\newblock \bibinfo{journal}{\emph{Information Processing \& Management}}
  \bibinfo{volume}{57}, \bibinfo{number}{6} (\bibinfo{year}{2020}),
  \bibinfo{pages}{102377}.
\newblock


\bibitem[\protect\citeauthoryear{FitzGerald and Hurst}{FitzGerald and
  Hurst}{2017}]%
        {fitzgerald2017implicit}
\bibfield{author}{\bibinfo{person}{Chlo{\"e} FitzGerald} {and}
  \bibinfo{person}{Samia Hurst}.} \bibinfo{year}{2017}\natexlab{}.
\newblock \showarticletitle{Implicit Bias in Healthcare Professionals: A
  Systematic Review}.
\newblock \bibinfo{journal}{\emph{BMC Medical Ethics}} \bibinfo{volume}{18},
  \bibinfo{number}{1} (\bibinfo{year}{2017}), \bibinfo{pages}{1--18}.
\newblock
\urldef\tempurl%
\url{https://bmcmedethics.biomedcentral.com/articles/10.1186/s12910-017-0179-8}
\showURL{%
\tempurl}


\bibitem[\protect\citeauthoryear{Hajian, Bonchi, and Castillo}{Hajian
  et~al\mbox{.}}{2016}]%
        {hajian2016algorithmic}
\bibfield{author}{\bibinfo{person}{Sara Hajian}, \bibinfo{person}{Francesco
  Bonchi}, {and} \bibinfo{person}{Carlos Castillo}.}
  \bibinfo{year}{2016}\natexlab{}.
\newblock \showarticletitle{Algorithmic Bias: From Discrimination Discovery to
  Fairness-Aware Data Mining}. In \bibinfo{booktitle}{\emph{Proceedings of the
  22nd ACM SIGKDD International Conference on Knowledge Discovery and Data
  Mining}} \emph{(\bibinfo{series}{KDD '16})}. \bibinfo{publisher}{Association
  for Computing Machinery}, \bibinfo{address}{New York, NY, USA},
  \bibinfo{pages}{2125–2126}.
\newblock
\showISBNx{9781450342322}
\urldef\tempurl%
\url{https://doi.org/10.1145/2939672.2945386}
\showDOI{\tempurl}


\bibitem[\protect\citeauthoryear{Jin, Wang, Zhang, Zheng, Ding, Xia, and
  Pan}{Jin et~al\mbox{.}}{2023}]%
        {JIN2023101906}
\bibfield{author}{\bibinfo{person}{Di Jin}, \bibinfo{person}{Luzhi Wang},
  \bibinfo{person}{He Zhang}, \bibinfo{person}{Yizhen Zheng},
  \bibinfo{person}{Weiping Ding}, \bibinfo{person}{Feng Xia}, {and}
  \bibinfo{person}{Shirui Pan}.} \bibinfo{year}{2023}\natexlab{}.
\newblock \showarticletitle{A Survey on Fairness-aware Recommender Systems}.
\newblock \bibinfo{journal}{\emph{Information Fusion}}  \bibinfo{volume}{100}
  (\bibinfo{year}{2023}), \bibinfo{pages}{101906}.
\newblock
\showISSN{1566-2535}
\urldef\tempurl%
\url{https://doi.org/10.1016/j.inffus.2023.101906}
\showDOI{\tempurl}


\bibitem[\protect\citeauthoryear{Joachims, Granka, Pan, Hembrooke, and
  Gay}{Joachims et~al\mbox{.}}{2005}]%
        {joachims2005accurately}
\bibfield{author}{\bibinfo{person}{Thorsten Joachims}, \bibinfo{person}{Laura
  Granka}, \bibinfo{person}{Bing Pan}, \bibinfo{person}{Helene Hembrooke},
  {and} \bibinfo{person}{Geri Gay}.} \bibinfo{year}{2005}\natexlab{}.
\newblock \showarticletitle{Accurately Interpreting Clickthrough Data as
  Implicit Feedback}. In \bibinfo{booktitle}{\emph{Proceedings of the 28th
  Annual International ACM SIGIR Conference on Research and Development in
  Information Retrieval}}. \bibinfo{publisher}{ACM}, \bibinfo{pages}{154--161}.
\newblock


\bibitem[\protect\citeauthoryear{Joachims, Swaminathan, and Schnabel}{Joachims
  et~al\mbox{.}}{2017}]%
        {joachims2017unbiased}
\bibfield{author}{\bibinfo{person}{Thorsten Joachims}, \bibinfo{person}{Adith
  Swaminathan}, {and} \bibinfo{person}{Tobias Schnabel}.}
  \bibinfo{year}{2017}\natexlab{}.
\newblock \showarticletitle{Unbiased Learning-to-Rank with Biased Feedback}. In
  \bibinfo{booktitle}{\emph{Proceedings of the Tenth ACM International
  Conference on Web Search and Data Mining}}. ACM, \bibinfo{pages}{781--789}.
\newblock


\bibitem[\protect\citeauthoryear{Jolls and Sunstein}{Jolls and
  Sunstein}{2006}]%
        {jolls2006law}
\bibfield{author}{\bibinfo{person}{Christine Jolls} {and}
  \bibinfo{person}{Cass~R. Sunstein}.} \bibinfo{year}{2006}\natexlab{}.
\newblock \showarticletitle{The Law of Implicit Bias}.
\newblock \bibinfo{journal}{\emph{Calif. L. Rev.}}  \bibinfo{volume}{94}
  (\bibinfo{year}{2006}), \bibinfo{pages}{969}.
\newblock


\bibitem[\protect\citeauthoryear{Karomat}{Karomat}{2023}]%
        {karomat2023unconscious}
\bibfield{author}{\bibinfo{person}{Kilicheva Karomat}.}
  \bibinfo{year}{2023}\natexlab{}.
\newblock \showarticletitle{Unconscious Bias in the Workplace}.
\newblock \bibinfo{journal}{\emph{International Journal on Integrated
  Education}} \bibinfo{volume}{6}, \bibinfo{number}{4} (\bibinfo{year}{2023}),
  \bibinfo{pages}{266--268}.
\newblock


\bibitem[\protect\citeauthoryear{Katariya, Kveton, Szepesvari, and
  Wen}{Katariya et~al\mbox{.}}{2016}]%
        {katariya2016dcm}
\bibfield{author}{\bibinfo{person}{Sumeet Katariya}, \bibinfo{person}{Branislav
  Kveton}, \bibinfo{person}{Csaba Szepesvari}, {and} \bibinfo{person}{Zheng
  Wen}.} \bibinfo{year}{2016}\natexlab{}.
\newblock \showarticletitle{{DCM} Bandits: Learning to Rank with Multiple
  Clicks}. In \bibinfo{booktitle}{\emph{International Conference on Machine
  Learning}}. \bibinfo{pages}{1215--1224}.
\newblock


\bibitem[\protect\citeauthoryear{Kay, Matuszek, and Munson}{Kay
  et~al\mbox{.}}{2015}]%
        {kay2015unequal}
\bibfield{author}{\bibinfo{person}{Matthew Kay}, \bibinfo{person}{Cynthia
  Matuszek}, {and} \bibinfo{person}{Sean~A. Munson}.}
  \bibinfo{year}{2015}\natexlab{}.
\newblock \showarticletitle{Unequal Representation and Gender Stereotypes in
  Image Search Results for Occupations}. In
  \bibinfo{booktitle}{\emph{Proceedings of the 33rd Annual ACM Conference on
  Human Factors in Computing Systems}} \emph{(\bibinfo{series}{CHI '15})}.
  \bibinfo{publisher}{Association for Computing Machinery},
  \bibinfo{address}{New York, NY, USA}, \bibinfo{pages}{3819–3828}.
\newblock
\showISBNx{9781450331456}
\urldef\tempurl%
\url{https://doi.org/10.1145/2702123.2702520}
\showDOI{\tempurl}


\bibitem[\protect\citeauthoryear{Kleinberg and Raghavan}{Kleinberg and
  Raghavan}{2018}]%
        {kleinberg2018selection}
\bibfield{author}{\bibinfo{person}{Jon Kleinberg} {and} \bibinfo{person}{Manish
  Raghavan}.} \bibinfo{year}{2018}\natexlab{}.
\newblock \showarticletitle{Selection Problems in the Presence of Implicit
  Bias}. In \bibinfo{booktitle}{\emph{9th Innovations in Theoretical Computer
  Science Conference (ITCS 2018)}}. Schloss Dagstuhl-Leibniz-Zentrum fuer
  Informatik.
\newblock


\bibitem[\protect\citeauthoryear{Krieg, Parada-Cabaleiro, Medicus, Lesota,
  Schedl, and Rekabsaz}{Krieg et~al\mbox{.}}{2022a}]%
        {krieg2022grep}
\bibfield{author}{\bibinfo{person}{Klara Krieg}, \bibinfo{person}{Emilia
  Parada-Cabaleiro}, \bibinfo{person}{Gertraud Medicus}, \bibinfo{person}{Oleg
  Lesota}, \bibinfo{person}{Markus Schedl}, {and} \bibinfo{person}{Navid
  Rekabsaz}.} \bibinfo{year}{2022}\natexlab{a}.
\newblock \showarticletitle{Grep-BiasIR: A Dataset for Investigating Gender
  Representation-Bias in Information Retrieval Results}.
\newblock \bibinfo{journal}{\emph{arXiv preprint arXiv:2201.07754}}
  (\bibinfo{year}{2022}).
\newblock


\bibitem[\protect\citeauthoryear{Krieg, Parada-Cabaleiro, Schedl, and
  Rekabsaz}{Krieg et~al\mbox{.}}{2022b}]%
        {krieg2022perceived}
\bibfield{author}{\bibinfo{person}{Klara Krieg}, \bibinfo{person}{Emilia
  Parada-Cabaleiro}, \bibinfo{person}{Markus Schedl}, {and}
  \bibinfo{person}{Navid Rekabsaz}.} \bibinfo{year}{2022}\natexlab{b}.
\newblock \showarticletitle{Do Perceived Gender Biases in Retrieval Results
  Affect Relevance Judgements?}
\newblock \bibinfo{journal}{\emph{arXiv preprint arXiv:2203.01731}}
  (\bibinfo{year}{2022}).
\newblock


\bibitem[\protect\citeauthoryear{Kumar}{Kumar}{2009}]%
        {kumar2009complaint}
\bibfield{author}{\bibinfo{person}{Rajeev Kumar}.}
  \bibinfo{year}{2009}\natexlab{}.
\newblock \bibinfo{title}{RTI Complaint}.
\newblock
\newblock
\newblock
\shownote{Decision No. CIC/SG/C/2009/001088/5392, Complaint No.
  CIC/SG/C/2009/001088.}


\bibitem[\protect\citeauthoryear{Kveton, Szepesvari, Wen, and Ashkan}{Kveton
  et~al\mbox{.}}{2015}]%
        {kveton2015cascading}
\bibfield{author}{\bibinfo{person}{Branislav Kveton}, \bibinfo{person}{Csaba
  Szepesvari}, \bibinfo{person}{Zheng Wen}, {and} \bibinfo{person}{Azin
  Ashkan}.} \bibinfo{year}{2015}\natexlab{}.
\newblock \showarticletitle{Cascading Bandits: Learning to Rank in the Cascade
  Model}. In \bibinfo{booktitle}{\emph{International Conference on Machine
  Learning}}. \bibinfo{pages}{767--776}.
\newblock


\bibitem[\protect\citeauthoryear{Lagr{\'e}e, Vernade, and Capp{\'e}}{Lagr{\'e}e
  et~al\mbox{.}}{2016}]%
        {lagree2016multiple}
\bibfield{author}{\bibinfo{person}{Paul Lagr{\'e}e}, \bibinfo{person}{Claire
  Vernade}, {and} \bibinfo{person}{Olivier Capp{\'e}}.}
  \bibinfo{year}{2016}\natexlab{}.
\newblock \showarticletitle{Multiple-play Bandits in the Position-based Model}.
  In \bibinfo{booktitle}{\emph{Advances in Neural Information Processing
  Systems}}. \bibinfo{pages}{1597--1605}.
\newblock


\bibitem[\protect\citeauthoryear{Lattimore and Szepesv{\'a}ri}{Lattimore and
  Szepesv{\'a}ri}{2020}]%
        {lattimore2019bandit}
\bibfield{author}{\bibinfo{person}{Tor Lattimore} {and} \bibinfo{person}{Csaba
  Szepesv{\'a}ri}.} \bibinfo{year}{2020}\natexlab{}.
\newblock \bibinfo{booktitle}{\emph{Bandit Algorithms}}.
\newblock \bibinfo{publisher}{Cambridge University Press}.
\newblock


\bibitem[\protect\citeauthoryear{Li, Kveton, Lattimore, Markov, de~Rijke,
  Szepesv{\'a}ri, and Zoghi}{Li et~al\mbox{.}}{2020}]%
        {li2020bubblerank}
\bibfield{author}{\bibinfo{person}{Chang Li}, \bibinfo{person}{Branislav
  Kveton}, \bibinfo{person}{Tor Lattimore}, \bibinfo{person}{Ilya Markov},
  \bibinfo{person}{Maarten de Rijke}, \bibinfo{person}{Csaba Szepesv{\'a}ri},
  {and} \bibinfo{person}{Masrour Zoghi}.} \bibinfo{year}{2020}\natexlab{}.
\newblock \showarticletitle{BubbleRank: Safe Online Learning to Re-rank via
  Implicit Click Feedback}. In \bibinfo{booktitle}{\emph{Uncertainty in
  Artificial Intelligence}}. PMLR, \bibinfo{pages}{196--206}.
\newblock


\bibitem[\protect\citeauthoryear{Liu, Rattan, and Savani}{Liu
  et~al\mbox{.}}{2023}]%
        {liu2023reducing}
\bibfield{author}{\bibinfo{person}{Zhi Liu}, \bibinfo{person}{Aneeta Rattan},
  {and} \bibinfo{person}{Krishna Savani}.} \bibinfo{year}{2023}\natexlab{}.
\newblock \showarticletitle{Reducing Gender Bias in the Evaluation and
  Selection of Future Leaders: The Role of Decision-makers' Mindsets about the
  Universality of Leadership Potential}.
\newblock \bibinfo{journal}{\emph{Journal of Applied Psychology}}
  (\bibinfo{year}{2023}).
\newblock
\urldef\tempurl%
\url{https://doi.org/10.1037/apl0001112}
\showDOI{\tempurl}


\bibitem[\protect\citeauthoryear{Massey~Jr}{Massey~Jr}{1951}]%
        {massey1951kolmogorov}
\bibfield{author}{\bibinfo{person}{Frank~J Massey~Jr}.}
  \bibinfo{year}{1951}\natexlab{}.
\newblock \showarticletitle{The Kolmogorov-Smirnov Test for Goodness of Fit}.
\newblock \bibinfo{journal}{\emph{J. Amer. Statist. Assoc.}}
  \bibinfo{volume}{46}, \bibinfo{number}{253} (\bibinfo{year}{1951}),
  \bibinfo{pages}{68--78}.
\newblock


\bibitem[\protect\citeauthoryear{Mehrotra, Pradelski, and Vishnoi}{Mehrotra
  et~al\mbox{.}}{2022}]%
        {mehrotra2022selection}
\bibfield{author}{\bibinfo{person}{Anay Mehrotra}, \bibinfo{person}{Bary S.~R.
  Pradelski}, {and} \bibinfo{person}{Nisheeth~K. Vishnoi}.}
  \bibinfo{year}{2022}\natexlab{}.
\newblock \showarticletitle{Selection in the Presence of Implicit Bias: The
  Advantage of Intersectional Constraints}. In \bibinfo{booktitle}{\emph{2022
  ACM Conference on Fairness, Accountability, and Transparency}}
  \emph{(\bibinfo{series}{FAccT '22})}. \bibinfo{publisher}{Association for
  Computing Machinery}, \bibinfo{address}{New York, NY, USA},
  \bibinfo{pages}{599–609}.
\newblock
\showISBNx{9781450393522}
\urldef\tempurl%
\url{https://doi.org/10.1145/3531146.3533124}
\showDOI{\tempurl}


\bibitem[\protect\citeauthoryear{Molenberghs}{Molenberghs}{2013}]%
        {Molenberghs2013the}
\bibfield{author}{\bibinfo{person}{Pascal Molenberghs}.}
  \bibinfo{year}{2013}\natexlab{}.
\newblock \showarticletitle{The Neuroscience of In-group Bias}.
\newblock \bibinfo{journal}{\emph{Neuroscience \& Biobehavioral Reviews}}
  \bibinfo{volume}{37}, \bibinfo{number}{8} (\bibinfo{year}{2013}),
  \bibinfo{pages}{1530--1536}.
\newblock
\showISSN{0149-7634}
\urldef\tempurl%
\url{https://doi.org/10.1016/j.neubiorev.2013.06.002}
\showDOI{\tempurl}


\bibitem[\protect\citeauthoryear{Morik, Singh, Hong, and Joachims}{Morik
  et~al\mbox{.}}{2020}]%
        {morik2020Controlling}
\bibfield{author}{\bibinfo{person}{Marco Morik}, \bibinfo{person}{Ashudeep
  Singh}, \bibinfo{person}{Jessica Hong}, {and} \bibinfo{person}{Thorsten
  Joachims}.} \bibinfo{year}{2020}\natexlab{}.
\newblock \showarticletitle{Controlling Fairness and Bias in Dynamic
  Learning-to-Rank}. In \bibinfo{booktitle}{\emph{Proceedings of the 43rd
  International ACM SIGIR Conference on Research and Development in Information
  Retrieval}} \emph{(\bibinfo{series}{SIGIR '20})}.
  \bibinfo{publisher}{Association for Computing Machinery},
  \bibinfo{address}{New York, NY, USA}, \bibinfo{pages}{429–438}.
\newblock
\showISBNx{9781450380164}
\urldef\tempurl%
\url{https://doi.org/10.1145/3397271.3401100}
\showDOI{\tempurl}


\bibitem[\protect\citeauthoryear{O'Sullivan, Kagabo, Prasad, Laporte, and
  Aiyer}{O'Sullivan et~al\mbox{.}}{2023}]%
        {o2023racial}
\bibfield{author}{\bibinfo{person}{Lucy O'Sullivan}, \bibinfo{person}{Whitney
  Kagabo}, \bibinfo{person}{Niyathi Prasad}, \bibinfo{person}{Dawn Laporte},
  {and} \bibinfo{person}{Amiethab Aiyer}.} \bibinfo{year}{2023}\natexlab{}.
\newblock \showarticletitle{Racial and Ethnic Bias in Medical School Clinical
  Grading: A Review}.
\newblock \bibinfo{journal}{\emph{Journal of surgical education}}
  (\bibinfo{year}{2023}).
\newblock


\bibitem[\protect\citeauthoryear{Pobrotyn, Bartczak, Synowiec, Bialobrzeski,
  and Bojar}{Pobrotyn et~al\mbox{.}}{2020}]%
        {Pobrotyn2020ContextAwareLT}
\bibfield{author}{\bibinfo{person}{Przemyslaw Pobrotyn},
  \bibinfo{person}{Tomasz Bartczak}, \bibinfo{person}{Mikolaj Synowiec},
  \bibinfo{person}{Radoslaw Bialobrzeski}, {and} \bibinfo{person}{Jaroslaw
  Bojar}.} \bibinfo{year}{2020}\natexlab{}.
\newblock \showarticletitle{Context-Aware Learning to Rank with
  Self-Attention}.
\newblock \bibinfo{journal}{\emph{arXiv preprint arXiv:2005.10084}}
  (\bibinfo{year}{2020}).
\newblock


\bibitem[\protect\citeauthoryear{Qin and Liu}{Qin and Liu}{2013}]%
        {qin2013introducing}
\bibfield{author}{\bibinfo{person}{Tao Qin} {and} \bibinfo{person}{Tie-Yan
  Liu}.} \bibinfo{year}{2013}\natexlab{}.
\newblock \showarticletitle{Introducing LETOR 4.0 Datasets}.
\newblock \bibinfo{journal}{\emph{arXiv preprint arXiv:1306.2597}}
  (\bibinfo{year}{2013}).
\newblock


\bibitem[\protect\citeauthoryear{Radlinski, Kleinberg, and Joachims}{Radlinski
  et~al\mbox{.}}{2008}]%
        {radlinski2008learning}
\bibfield{author}{\bibinfo{person}{Filip Radlinski}, \bibinfo{person}{Robert
  Kleinberg}, {and} \bibinfo{person}{Thorsten Joachims}.}
  \bibinfo{year}{2008}\natexlab{}.
\newblock \showarticletitle{Learning Diverse Rankings with Multi-armed
  Bandits}. In \bibinfo{booktitle}{\emph{Proceedings of the 25th International
  Conference on Machine Learning}}. \bibinfo{pages}{784--791}.
\newblock


\bibitem[\protect\citeauthoryear{Raj and Ekstrand}{Raj and Ekstrand}{2022}]%
        {raj2022measuring}
\bibfield{author}{\bibinfo{person}{Amifa Raj} {and} \bibinfo{person}{Michael~D.
  Ekstrand}.} \bibinfo{year}{2022}\natexlab{}.
\newblock \showarticletitle{Measuring Fairness in Ranked Results: An Analytical
  and Empirical Comparison}. In \bibinfo{booktitle}{\emph{Proceedings of the
  45th International ACM SIGIR Conference on Research and Development in
  Information Retrieval}} \emph{(\bibinfo{series}{SIGIR '22})}.
  \bibinfo{publisher}{Association for Computing Machinery},
  \bibinfo{address}{New York, NY, USA}, \bibinfo{pages}{726–736}.
\newblock
\showISBNx{9781450387323}
\urldef\tempurl%
\url{https://doi.org/10.1145/3477495.3532018}
\showDOI{\tempurl}


\bibitem[\protect\citeauthoryear{Saito and Joachims}{Saito and
  Joachims}{2022}]%
        {saito2022fair}
\bibfield{author}{\bibinfo{person}{Yuta Saito} {and} \bibinfo{person}{Thorsten
  Joachims}.} \bibinfo{year}{2022}\natexlab{}.
\newblock \showarticletitle{Fair Ranking as Fair Division: Impact-Based
  Individual Fairness in Ranking}. In \bibinfo{booktitle}{\emph{Proceedings of
  the 28th ACM SIGKDD Conference on Knowledge Discovery and Data Mining}}
  \emph{(\bibinfo{series}{KDD '22})}. \bibinfo{publisher}{Association for
  Computing Machinery}, \bibinfo{address}{New York, NY, USA},
  \bibinfo{pages}{1514–1524}.
\newblock
\showISBNx{9781450393850}
\urldef\tempurl%
\url{https://doi.org/10.1145/3534678.3539353}
\showDOI{\tempurl}


\bibitem[\protect\citeauthoryear{Sarvi, Heuss, Aliannejadi, Schelter, and
  de~Rijke}{Sarvi et~al\mbox{.}}{2022}]%
        {sarvi2021understanding}
\bibfield{author}{\bibinfo{person}{Fatemeh Sarvi}, \bibinfo{person}{Maria
  Heuss}, \bibinfo{person}{Mohammad Aliannejadi}, \bibinfo{person}{Sebastian
  Schelter}, {and} \bibinfo{person}{Maarten de Rijke}.}
  \bibinfo{year}{2022}\natexlab{}.
\newblock \showarticletitle{Understanding and Mitigating the Effect of Outliers
  in Fair Ranking}. In \bibinfo{booktitle}{\emph{WSDM 2022: The Fifteenth
  International Conference on Web Search and Data Mining}}.
  \bibinfo{publisher}{ACM}, \bibinfo{pages}{861--869}.
\newblock


\bibitem[\protect\citeauthoryear{Saxena and Jain}{Saxena and Jain}{2021}]%
        {saxena2021exploring}
\bibfield{author}{\bibinfo{person}{Shrikant Saxena} {and}
  \bibinfo{person}{Shweta Jain}.} \bibinfo{year}{2021}\natexlab{}.
\newblock \showarticletitle{Exploring and Mitigating Gender Bias in Recommender
  Systems with Explicit Feedback}.
\newblock \bibinfo{journal}{\emph{arXiv preprint arXiv:2112.02530}}
  (\bibinfo{year}{2021}).
\newblock


\bibitem[\protect\citeauthoryear{Singh and Joachims}{Singh and
  Joachims}{2018}]%
        {singh-2018-fairness}
\bibfield{author}{\bibinfo{person}{Ashudeep Singh} {and}
  \bibinfo{person}{Thorsten Joachims}.} \bibinfo{year}{2018}\natexlab{}.
\newblock \showarticletitle{Fairness of Exposure in Rankings}. In
  \bibinfo{booktitle}{\emph{Proceedings of the 24th ACM SIGKDD International
  Conference on Knowledge Discovery \& Data Mining}}
  \emph{(\bibinfo{series}{KDD '18})}. \bibinfo{publisher}{Association for
  Computing Machinery}, \bibinfo{address}{New York, NY, USA},
  \bibinfo{pages}{2219–2228}.
\newblock
\showISBNx{9781450355520}
\urldef\tempurl%
\url{https://doi.org/10.1145/3219819.3220088}
\showDOI{\tempurl}


\bibitem[\protect\citeauthoryear{Singh and Joachims}{Singh and
  Joachims}{2019}]%
        {singh2019policy}
\bibfield{author}{\bibinfo{person}{Ashudeep Singh} {and}
  \bibinfo{person}{Thorsten Joachims}.} \bibinfo{year}{2019}\natexlab{}.
\newblock \showarticletitle{Policy Learning for Fairness in Ranking}. In
  \bibinfo{booktitle}{\emph{Advances in Neural Information Processing
  Systems}}, \bibfield{editor}{\bibinfo{person}{H.~Wallach},
  \bibinfo{person}{H.~Larochelle}, \bibinfo{person}{A.~Beygelzimer},
  \bibinfo{person}{F.~d\textquotesingle Alch\'{e}-Buc},
  \bibinfo{person}{E.~Fox}, {and} \bibinfo{person}{R.~Garnett}} (Eds.),
  Vol.~\bibinfo{volume}{32}. \bibinfo{publisher}{Curran Associates, Inc.}
\newblock
\urldef\tempurl%
\url{https://proceedings.neurips.cc/paper_files/paper/2019/file/9e82757e9a1c12cb710ad680db11f6f1-Paper.pdf}
\showURL{%
\tempurl}


\bibitem[\protect\citeauthoryear{S\"{u}hr, Hilgard, and Lakkaraju}{S\"{u}hr
  et~al\mbox{.}}{2021}]%
        {suhr2021does}
\bibfield{author}{\bibinfo{person}{Tom S\"{u}hr}, \bibinfo{person}{Sophie
  Hilgard}, {and} \bibinfo{person}{Himabindu Lakkaraju}.}
  \bibinfo{year}{2021}\natexlab{}.
\newblock \showarticletitle{Does Fair Ranking Improve Minority Outcomes?
  Understanding the Interplay of Human and Algorithmic Biases in Online Hiring}
  \emph{(\bibinfo{series}{AIES '21})}. \bibinfo{publisher}{Association for
  Computing Machinery}, \bibinfo{address}{New York, NY, USA},
  \bibinfo{pages}{989–999}.
\newblock
\showISBNx{9781450384735}
\urldef\tempurl%
\url{https://doi.org/10.1145/3461702.3462602}
\showDOI{\tempurl}


\bibitem[\protect\citeauthoryear{Vardasbi, de~Rijke, and Markov}{Vardasbi
  et~al\mbox{.}}{2021}]%
        {vardasbi2021mixture}
\bibfield{author}{\bibinfo{person}{Ali Vardasbi}, \bibinfo{person}{Maarten de
  Rijke}, {and} \bibinfo{person}{Ilya Markov}.}
  \bibinfo{year}{2021}\natexlab{}.
\newblock \showarticletitle{Mixture-Based Correction for Position and Trust
  Bias in Counterfactual Learning to Rank}. In
  \bibinfo{booktitle}{\emph{Proceedings of the 30th ACM International
  Conference on Information \& Knowledge Management}}
  \emph{(\bibinfo{series}{CIKM '21})}. \bibinfo{publisher}{Association for
  Computing Machinery}, \bibinfo{address}{New York, NY, USA},
  \bibinfo{pages}{1869–1878}.
\newblock
\showISBNx{9781450384469}
\urldef\tempurl%
\url{https://doi.org/10.1145/3459637.3482275}
\showDOI{\tempurl}


\bibitem[\protect\citeauthoryear{Vardasbi, Oosterhuis, and de~Rijke}{Vardasbi
  et~al\mbox{.}}{2020}]%
        {vardasbi2020when}
\bibfield{author}{\bibinfo{person}{Ali Vardasbi}, \bibinfo{person}{Harrie
  Oosterhuis}, {and} \bibinfo{person}{Maarten de Rijke}.}
  \bibinfo{year}{2020}\natexlab{}.
\newblock \showarticletitle{When Inverse Propensity Scoring Does Not Work:
  Affine Corrections for Unbiased Learning to Rank}. In
  \bibinfo{booktitle}{\emph{Proceedings of the 29th ACM International
  Conference on Information \& Knowledge Management}}
  \emph{(\bibinfo{series}{CIKM '20})}. \bibinfo{publisher}{Association for
  Computing Machinery}, \bibinfo{address}{New York, NY, USA},
  \bibinfo{pages}{1475–1484}.
\newblock
\showISBNx{9781450368599}
\urldef\tempurl%
\url{https://doi.org/10.1145/3340531.3412031}
\showDOI{\tempurl}


\bibitem[\protect\citeauthoryear{Vardasbi, Sarvi, and de~Rijke}{Vardasbi
  et~al\mbox{.}}{2022}]%
        {vardasbi2022probabilistic}
\bibfield{author}{\bibinfo{person}{Ali Vardasbi}, \bibinfo{person}{Fatemeh
  Sarvi}, {and} \bibinfo{person}{Maarten de Rijke}.}
  \bibinfo{year}{2022}\natexlab{}.
\newblock \showarticletitle{Probabilistic Permutation Graph Search: Black-Box
  Optimization for Fairness in Ranking}. In
  \bibinfo{booktitle}{\emph{Proceedings of the 45th International ACM SIGIR
  Conference on Research and Development in Information Retrieval}}
  \emph{(\bibinfo{series}{SIGIR '22})}. \bibinfo{publisher}{Association for
  Computing Machinery}, \bibinfo{address}{New York, NY, USA},
  \bibinfo{pages}{715–725}.
\newblock
\showISBNx{9781450387323}
\urldef\tempurl%
\url{https://doi.org/10.1145/3477495.3532045}
\showDOI{\tempurl}


\bibitem[\protect\citeauthoryear{Vlasceanu and Amodio}{Vlasceanu and
  Amodio}{2022}]%
        {vlasceanu2022propagation}
\bibfield{author}{\bibinfo{person}{Madalina Vlasceanu} {and}
  \bibinfo{person}{David~M. Amodio}.} \bibinfo{year}{2022}\natexlab{}.
\newblock \showarticletitle{Propagation of Societal Gender Inequality by
  Internet Search Algorithms}.
\newblock \bibinfo{journal}{\emph{Proceedings of the National Academy of
  Sciences}} \bibinfo{volume}{119}, \bibinfo{number}{29}
  (\bibinfo{year}{2022}), \bibinfo{pages}{e2204529119}.
\newblock
\urldef\tempurl%
\url{https://doi.org/10.1073/pnas.2204529119}
\showDOI{\tempurl}
\showeprint{https://www.pnas.org/doi/pdf/10.1073/pnas.2204529119}


\bibitem[\protect\citeauthoryear{Wang, Golbandi, Bendersky, Metzler, and
  Najork}{Wang et~al\mbox{.}}{2018}]%
        {wang2018position}
\bibfield{author}{\bibinfo{person}{Xuanhui Wang}, \bibinfo{person}{Nadav
  Golbandi}, \bibinfo{person}{Michael Bendersky}, \bibinfo{person}{Donald
  Metzler}, {and} \bibinfo{person}{Marc Najork}.}
  \bibinfo{year}{2018}\natexlab{}.
\newblock \showarticletitle{Position Bias Estimation for Unbiased Learning to
  Rank in Personal Search}. In \bibinfo{booktitle}{\emph{Proceedings of the
  Eleventh ACM International Conference on Web Search and Data Mining}}. ACM,
  \bibinfo{pages}{610--618}.
\newblock


\bibitem[\protect\citeauthoryear{Wu, Mitra, Ma, Diaz, and Liu}{Wu
  et~al\mbox{.}}{2022}]%
        {wu2022joint}
\bibfield{author}{\bibinfo{person}{Haolun Wu}, \bibinfo{person}{Bhaskar Mitra},
  \bibinfo{person}{Chen Ma}, \bibinfo{person}{Fernando Diaz}, {and}
  \bibinfo{person}{Xue Liu}.} \bibinfo{year}{2022}\natexlab{}.
\newblock \showarticletitle{Joint Multisided Exposure Fairness for
  Recommendation}. In \bibinfo{booktitle}{\emph{Proceedings of the 45th
  International ACM SIGIR Conference on Research and Development in Information
  Retrieval}} \emph{(\bibinfo{series}{SIGIR '22})}.
  \bibinfo{publisher}{Association for Computing Machinery},
  \bibinfo{address}{New York, NY, USA}, \bibinfo{pages}{703–714}.
\newblock
\showISBNx{9781450387323}
\urldef\tempurl%
\url{https://doi.org/10.1145/3477495.3532007}
\showDOI{\tempurl}


\bibitem[\protect\citeauthoryear{Yadav, Du, and Joachims}{Yadav
  et~al\mbox{.}}{2021}]%
        {yadav2021policy}
\bibfield{author}{\bibinfo{person}{Himank Yadav}, \bibinfo{person}{Zhengxiao
  Du}, {and} \bibinfo{person}{Thorsten Joachims}.}
  \bibinfo{year}{2021}\natexlab{}.
\newblock \showarticletitle{Policy-Gradient Training of Fair and Unbiased
  Ranking Functions}. In \bibinfo{booktitle}{\emph{Proceedings of the 44th
  International ACM SIGIR Conference on Research and Development in Information
  Retrieval}}. \bibinfo{publisher}{Association for Computing Machinery},
  \bibinfo{address}{New York, NY, USA}, \bibinfo{pages}{1044--1053}.
\newblock
\showISBNx{9781450380379}
\urldef\tempurl%
\url{https://doi.org/10.1145/3404835.3462953}
\showURL{%
\tempurl}


\bibitem[\protect\citeauthoryear{Zdaniuk and Levine}{Zdaniuk and
  Levine}{2001}]%
        {Zdaniuk2001group}
\bibfield{author}{\bibinfo{person}{Bozena Zdaniuk} {and}
  \bibinfo{person}{John~M. Levine}.} \bibinfo{year}{2001}\natexlab{}.
\newblock \showarticletitle{Group Loyalty: Impact of Members' Identification
  and Contributions}.
\newblock \bibinfo{journal}{\emph{Journal of Experimental Social Psychology}}
  \bibinfo{volume}{37}, \bibinfo{number}{6} (\bibinfo{year}{2001}),
  \bibinfo{pages}{502--509}.
\newblock
\showISSN{0022-1031}
\urldef\tempurl%
\url{https://doi.org/10.1006/jesp.2000.1474}
\showDOI{\tempurl}


\bibitem[\protect\citeauthoryear{Zoghi, Tunys, Li, Jose, Chen, Chin, and
  de~Rijke}{Zoghi et~al\mbox{.}}{2016}]%
        {zoghi2016click}
\bibfield{author}{\bibinfo{person}{Masrour Zoghi},
  \bibinfo{person}{Tom{\'a}{\v{s}} Tunys}, \bibinfo{person}{Lihong Li},
  \bibinfo{person}{Damien Jose}, \bibinfo{person}{Junyan Chen},
  \bibinfo{person}{Chun~Ming Chin}, {and} \bibinfo{person}{Maarten de Rijke}.}
  \bibinfo{year}{2016}\natexlab{}.
\newblock \showarticletitle{Click-based Hot Fixes for Underperforming Torso
  Queries}. In \bibinfo{booktitle}{\emph{Proceedings of the 39th International
  ACM SIGIR conference on Research and Development in Information Retrieval}}.
  \bibinfo{pages}{195--204}.
\newblock


\end{thebibliography}
